\documentclass[acmsmall, authorversion]{acmart}
\usepackage{etoolbox}
\usepackage[shortcuts]{extdash}

\usepackage{booktabs}
\usepackage{algorithm}
\usepackage[noend]{algpseudocode}

\usepackage{mathtools}
\usepackage{graphicx}
\usepackage{dblfloatfix}
\usepackage{fixfoot}
\usepackage[shortlabels]{enumitem}
\usepackage{subcaption}

\usepackage{amsthm}

\usepackage[utf8]{inputenc}

\newtoggle{widecol}
\toggletrue{widecol}
\newcommand{\ifwidecol}{\iftoggle{widecol}}
\newcommand{\ifwidecolignorealigny}{%
  \ifwidecol{%
    \renewcommand{\\}{}
    \renewcommand{\notag}{}}{}
    \renewcommand{\nonumber}{}{}
    \renewcommand{\MoveEqLeft}{}}

\newcommand{\policyname}[1]{\ifmmode\textnormal{#1}\else\mbox{#1}\fi}
\def\P/{\policyname{P}}
\def\Gp/{\policyname{G\=/P}}
\def\Prio/{\policyname{Priority\=/\textit{c}}}
\def\Srpt/{\policyname{SRPT}}
\def\Psjf/{\policyname{PSJF}}
\def\GpPrio/{\policyname{\Gp//\Prio/}}
\def\MgPrio/{\policyname{M/G/1/\Prio/}}
\def\GpSrpt/{\policyname{\Gp//\Srpt/}}
\def\ArbArb/{\policyname{P$'$/S$'$}}

\renewcommand{\bar}{\overline}

\newcommand{\st}{\textnormal{st}}

\DeclareFixedFootnote{\matuszewska}{See Appendix~\ref{app:matuszewska}.}

\makeatletter
\newtheorem*{rep@theorem}{\rep@title}
\newcommand{\newreptheorem}[2]{%
\newenvironment{rep#1}[1]{%
 \def\rep@title{#2 \ref{##1}}%
 \begin{rep@theorem}}%
 {\end{rep@theorem}}}

\renewcommand{\fs@ruled}{%
  \def\@fs@cfont{\bfseries}\let\@fs@capt\floatc@ruled
  \def\@fs@pre{%
    \vspace{\abovetopsep}%
    \hrule height\heavyrulewidth \kern2pt%
    \vspace{\belowrulesep}}%
  \def\@fs@post{%
    \vspace{\aboverulesep}%
    \kern2pt\hrule height\heavyrulewidth \relax%
    \vspace{\belowbottomsep}}%
  \def\@fs@mid{%
    \vspace{\aboverulesep}%
    \kern2pt\hrule height\lightrulewidth \kern2pt%
    \vspace{\belowrulesep}}%
  \let\@fs@iftopcapt\iftrue}

\newtheoremstyle{acmcase}
  {.5\baselineskip\@plus.2\baselineskip
    \@minus.2\baselineskip}
  {.5\baselineskip\@plus.2\baselineskip
    \@minus.2\baselineskip}
  {\@acmdefinitionbodyfont}
  {\@acmdefinitionindent}
  {\@acmdefinitionheadfont}
  {.}
  {.5em}
  {\thmname{#1}\thmnumber{ #2}\thmnote{{\@acmdefinitionnotefont: #3}}}
\theoremstyle{acmcase}
\newtheorem{case}{Case}
\makeatother

\makeatletter
\@ifundefined{theorem}{%
\newtheorem{theorem}{Theorem}%
}{}
\@ifundefined{corollary}{%
\newtheorem{corollary}{Corollary}%
}{}
\@ifundefined{lemma}{%
\newtheorem{lemma}{Lemma}%
}{}
\@ifundefined{definition}{%
\theoremstyle{acmdefinition}
\newtheorem{definition}{Definition}%
}{}
\@ifundefined{remark}{%
\theoremstyle{remark}
\newtheorem{remark}{Remark}%
}{}
\makeatother

\newreptheorem{theorem}{Theorem}
\newreptheorem{corollary}{Corollary}
\newreptheorem{lemma}{Lemma}

\DeclarePairedDelimiter{\floor}{\lfloor}{\rfloor}

\newcommand{\ldelimE}{[}
\newcommand{\rdelimE}{]}
\DeclarePairedDelimiterXPP{\E}[1]{\mathbf{E}}{\ldelimE}{\rdelimE}{}{%
  
  #1}

\renewcommand{\Pr}{\mathbf{P}}

\title[Load Balancing Guardrails]{%
  Load Balancing Guardrails:
  Keeping Your Heavy Traffic on the Road to Low Response Times}

\author{Isaac Grosof}
\affiliation{%
  \institution{Carnegie Mellon University}
  \department{Computer Science Department}
  \streetaddress{5000 Forbes Ave}
  \city{Pittsburgh}
  \state{PA}
  \postcode{15213}
  \country{USA}}
\email{igrosof@cs.cmu.edu}

\author{Ziv Scully}
\affiliation{%
  \institution{Carnegie Mellon University}
  \department{Computer Science Department}
  \streetaddress{5000 Forbes Ave}
  \city{Pittsburgh}
  \state{PA}
  \postcode{15213}
  \country{USA}}
\email{zscully@cs.cmu.edu}

\author{Mor Harchol-Balter}
\affiliation{%
  \institution{Carnegie Mellon University}
  \department{Computer Science Department}
  \streetaddress{5000 Forbes Ave}
  \city{Pittsburgh}
  \state{PA}
  \postcode{15213}
  \country{USA}}
\email{harchol@cs.cmu.edu}

\authorsaddresses{%
  Authors' address:
  Carnegie Mellon University,
  Computer Science Department,
  5000 Forbes Ave,
  Pittsburgh,
  PA 15213,
  USA,
  $\{\text{igrosof}, \text{zscully}, \text{harchol}\}$@cs.cmu.edu.}

\begin{document}

\setcopyright{acmcopyright}
\acmJournal{POMACS}
\acmYear{2019}
\acmVolume{3}
\acmNumber{2}
\acmArticle{42}
\acmMonth{6}
\acmPrice{15.00}
\acmDOI{10.1145/3326157}

\makeatletter
\if@ACM@authorversion
  \let\@oldfntcp\footnotetextcopyrightpermission
  \newlength{\authorversionskip}
  \setlength{\authorversionskip}{0.39in}
  \renewcommand{\footnotetextcopyrightpermission}[1]{%
    \@oldfntcp{#1\vspace{\authorversionskip}}}
  \acmArticleSeq{0}
  \let\@oldhffont\@headfootfont
  \def\@headfootfont{%
    \renewcommand\@acmArticle[2]{##1}
    \@oldhffont}
\fi
\makeatother




\begin{abstract}
  Load balancing systems,
    comprising a central dispatcher and a scheduling policy at each server,
    are widely used in practice,
    and their response time has been extensively studied in the theoretical literature.
    While much is known about the scenario where the scheduling at the servers
  is First-Come-First-Served (FCFS),
  to minimize mean response time
  we must use Shortest-Remaining-Processing-Time (\Srpt/) scheduling at the servers.
  Much less is known about dispatching polices when \Srpt/ scheduling is used.
    Unfortunately, traditional dispatching policies
    that are used in practice in systems with FCFS servers
    often have poor performance
    in systems with \Srpt/ servers.
    In this paper, we devise a simple fix that can be applied
    to any dispatching policy.
    This fix, called \emph{guardrails},
    ensures that the dispatching policy yields optimal mean response time
    under heavy traffic
    when used in a system with \Srpt/ servers.
    \emph{Any} dispatching policy,
    when augmented with guardrails,
    becomes heavy-traffic optimal.
    Our results yield the first analytical bounds
    on mean response time for load balancing systems
    with \Srpt/ scheduling at the servers.
\end{abstract}
\maketitle
\section{Introduction}
\label{sec:intro}

Load balancers are ubiquitous throughout computer systems.
They act as a front-end to web server farms,
distributing HTTP requests to different servers.
They likewise act as a front-end to data centers and cloud computing pools,
where they distribute requests among servers and virtual machines.

\begin{figure}
  \includegraphics[width=\ifwidecol{0.55}{}\linewidth]{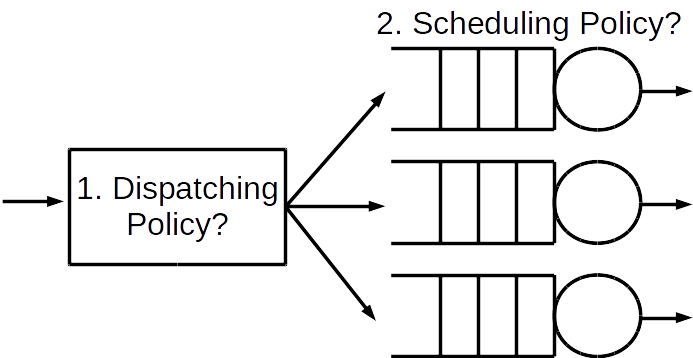}
  \caption{Two decision points within a load balancing system:
  (1) Pick the dispatching policy. (2) Pick the scheduling policy for the servers.}
  \label{fig:load-balancing-diagram}
\end{figure}

In this paper, we consider the immediate dispatch
load balancing model,
where each arriving job is immediately dispatched to a server,
as shown in Figure~\ref{fig:load-balancing-diagram}.
The system has two decision points:
\begin{enumerate}[(1)]
\item
  A \emph{dispatching policy} decides how to distribute jobs across the servers.
\item
  A \emph{scheduling policy} at each server
  decides which job to serve among those at that server.
\end{enumerate}
We ask:
\begin{quote}
  What (1) dispatching policy and (2) scheduling policy
  should we use to \emph{minimize mean response time} of jobs?
\end{quote}
We assume that the job arrival process is Poisson
and that job sizes are i.i.d. from a general size distribution.
We assume jobs are preemptible with no loss of work.
Finally, we assume that job sizes are known at the time the job arrives in the system.

With these assumptions, the scheduling question turns out to be easy to answer:
use Shortest-Remaining-Processing-Time (\Srpt/) at the servers.
No matter what dispatching decisions are made,
if we consider the sequence of jobs dispatched to a particular server,
the policy which minimizes mean response time for that server
must be to schedule those jobs in \Srpt/ order.
This follows from from the optimality of \Srpt/
for arbitrary arrival sequences~\cite{schrage-srpt-optimal}.
\Srpt/ scheduling is in fact already used in backend servers
\cite{ousterhout-homa, mor-web-server}.
Thus, in the remainder of this paper we assume \Srpt/ is used at the servers.

The question remains:
What dispatching policy minimizes mean response time
given \Srpt/ service at the servers?
While many dispatching policies have been considered in the literature,
they have mostly been considered in the context of
First-Come-First-Served (FCFS)
or Processor-Sharing (PS)
scheduling at the servers.
Popular dispatching policies include
Random routing \cite{963420, mor-sita},
Least-Work-Left (LWL) \cite{Harchol-Balter:2009:SRT:1555349.1555383, mor-sita, Bramson2012},
Join-Shortest-Queue (JSQ) \cite{winston_1977, weber_1978, 55688, GUPTA20071062},
JSQ\=/$d$ \cite{Bramson2012, Mukherjee:2016:UPL:3003977.3003990, 963420, Li2011},
Size-Interval-Task-Assignment (SITA) \cite{mor-sita, Bachmat:2008:ASI:1453175.1453199, FENG2005475},
Round-Robin (RR) \cite{liu-round-robin, mor-sita},
and many more \cite{Altman2011, Bonald:2004:ILB:1005686.1005729, ZHOU2018176, sbrc}.
However, only the simplest of these policies,
such as Random and RR,
have been studied for \Srpt/ servers \cite{mor-book, Down2006}.

One might hope that the same policies that yield low mean response time
when servers use FCFS scheduling
would also perform well when servers use \Srpt/ scheduling.
Unfortunately, this does not always hold.
For example, when the servers use FCFS,
it is well-known that LWL dispatching,
which sends each job to the server with the least remaining work,
outperforms Random dispatching,
which sends each job to a randomly chosen server.
(We write this as LWL/FCFS outperforms Random/FCFS.)
However, the opposite can happen when the servers use \Srpt/:
as shown in the scenario in Figure~\ref{fig:srpt-vs-fcfs},
Random/\Srpt/ can outperform LWL/\Srpt/.
Moreover, the performance difference is highly significant:
Random/\Srpt/ outperforms LWL/\Srpt/ by a factor of 5 or more under heavy load.
This means that LWL is making serious mistakes in dispatching decisions.
We can therefore see that the heuristics that served us well
for FCFS servers can steer us awry when we use \Srpt/ servers.

In this paper, we introduce \emph{guardrails},
a new technique for creating dispatching policies.
Given an arbitrary dispatching policy~\P/,
applying guardrails results in an improved policy Guarded-P (\Gp/).
We prove that the improved policy \Gp/ has asymptotically optimal mean response time
in the heavy traffic limit,
no matter what the initial policy~\P/ is.
We also show empirically that
adding guardrails to a policy almost always decreases its mean response time
(and never significantly increases it),
even outside the heavy-traffic regime.

\begin{figure}
  \includegraphics[width=\ifwidecol{0.55}{}\linewidth, trim={0 0 0 0.5em}]{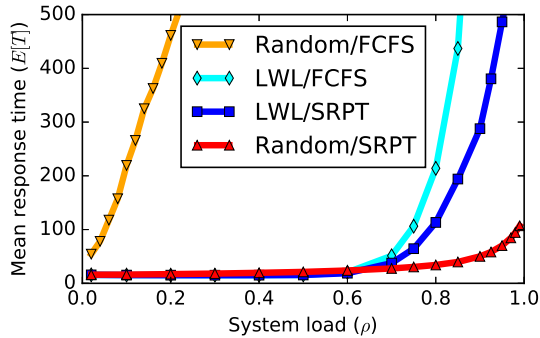}
  \vspace{\ifwidecol{-1.25}{-2.5}\baselineskip}
  \caption{Two dispatching policies: Random and LWL.
    Two scheduling policies: FCFS and \Srpt/.
    FCFS scheduling at the servers
    yields higher mean response time as a function of load,
    compared with \Srpt/ scheduling at the servers.
    Random dispatching is worse than LWL dispatching
    under FCFS scheduling at the servers,
    but Random dispatching is better than LWL dispatching
    under \Srpt/ scheduling at the servers.
    Simulation uses $k= 10$ servers. Size distribution shown is Bimodal with
    jobs of size 1 with probability 99.95\%
    and jobs of size 1000 with probability 0.05\%.}
  \label{fig:srpt-vs-fcfs}
\end{figure}

\begin{figure}
  \includegraphics[width=\ifwidecol{0.55}{}\linewidth, trim={0 0 0 0.5em}]{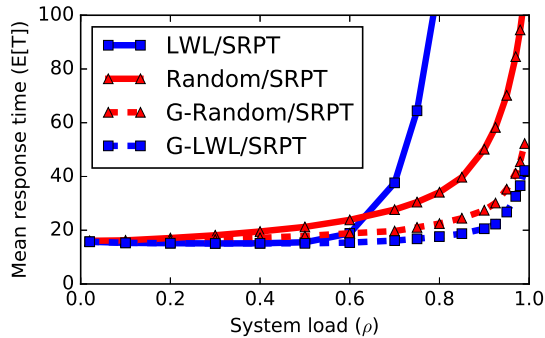}
  \vspace{\ifwidecol{-1.25}{-2.5}\baselineskip}
  \caption{Adding guardrails to LWL yields much lower mean response time
    as a function of load. Guardrails yield a factor of 3 improvement
    even at $\rho=0.8$, and a factor of 7 improvement at $\rho=0.9$.
    Adding guardrails to Random also yields significantly
    lower mean response time as a function of load.
    Simulation uses $k=10$ servers. Size distribution shown is Bimodal with
    jobs of size 1 with probability 99.95\%
    and jobs of size 1000 with probability 0.05\%.
    The guardrails have tightness $g=2$.}
  \label{fig:intro-guardrails}
\end{figure}

As an example of the power of guardrails,
Figure~\ref{fig:intro-guardrails} shows
the performance of guarded versions of LWL and Random,
namely G-LWL and G-Random.
The guardrails stop LWL from making serious mistakes
and dramatically improve its performance.
Random dispatching also benefits from guardrails.
Moreover, the guarded policies have a theoretical guarantee:
In the limit as load $\rho \to 1$,
G-Random/\Srpt/ and G-LWL/\Srpt/ converge to the optimal mean response time.
In contrast, unguarded Random/\Srpt/ is a factor of $k$ worse than optimal
in the $\rho \to 1$ limit,
where $k$ is the number of servers.

This paper makes the following contributions:
\begin{itemize}
\item In Section~\ref{sec:guardrails},
  we introduce guardrails,
  a technique for improving any dispatching policy.
\item In Section~\ref{sec:theorems},
  we bound the mean response time of any guarded dispatching policy
  when paired with SRPT scheduling at the servers.
  Using that bound,
  we prove that any guarded policy has asymptotically optimal mean response time
  as load $\rho \to 1$,
  subject to a technical condition on the job size distribution
  roughly equivalent to finite variance.
\item In Section~\ref{sec:simulation},
  we consider a wide variety of common dispatching policies.
  We empirically show that guardrails improve most of these at all loads.
\item In Section~\ref{sec:practical},
  we discuss practical considerations and extensions of guardrails,
  including
  \begin{itemize}
  \item guardrails for large systems,
    which may have multiple dispatchers and network delays;
  \item guardrails for scheduling policies other than \Srpt/; and
  \item guardrails for heterogeneous servers.
  \end{itemize}
\end{itemize}
We give a more technical summary of our theoretical results
and review related work in Section~\ref{sec:technical-summary}.

\section{Load Balancing Guardrails}
\label{sec:guardrails}

\subsection{What are Guardrails?}
\label{sec:guard-def}

Traditional dispatching policies aim to equalize load at each server.
However, minimizing mean response time requires more than balancing load:
we also need to find a way to favor small jobs.
Given that every server uses \Srpt/ scheduling,
if we can spread out the small jobs across the servers,
then we ensure that the maximum possible number of servers are working
on the smallest jobs available.
Our idea is to take any dispatching policy
and add ``guardrails'' that force it to spread out small jobs across the servers.

In the discussion above, ``small'' is a relative term.
Whatever the size of the smallest jobs currently in the system,
we would like to spread out jobs near that size across the servers.
To do this, we define the \emph{rank} of a job of size~$x$ to be
\begin{equation}
  \label{eq:rank}
  r = \floor{\log_c x},
\end{equation}
where $c > 1$ is a constant called the \emph{guardrail rank width}
(see Section~\ref{sec:choosing-c}).
The idea of guardrails is to
spread out the jobs within a rank~$r$ across the servers,
doing so separately for each rank~$r$.
To do so, for each rank~$r$ and each server~$s$, the dispatcher stores a \emph{guardrail work counter}~$G^r_s$.
When dispatching a job of size~$x$ to server~$s$,
the dispatcher increases $G^r_s$ by~$x$, with $r$ given by \eqref{eq:rank}.\footnote{%
  The dispatcher also occasionally decreases work counters,
  as explained in Section~\ref{sec:resets}.}
Guardrails are a set of constraints
which ensure that no two rank~$r$ work counters are ever too far apart.

\begin{definition}
  \label{def:guardrails}
  A dispatching policy \emph{satisfies guardrails with tightness~$g$}
  if at all times
  \begin{equation*}
    |G^r_s - G^r_{s'}| \leq gc^{r+1}
  \end{equation*}
  for all ranks~$r$ and all pairs of servers $s$ and~$s'$,
  where $c > 1$ is the same constant as in \eqref{eq:rank}.
  The tightness can be any constant $g \geq 1$.
\end{definition}

We sometimes say that a particular dispatching decision
satisfies (respectively, violates) guardrails
if it satisfies (respectively, violates)
the constraints imposed by Definition~\ref{def:guardrails}.

\subsubsection{Choosing the Guardrail Rank Width~$c$}
\label{sec:choosing-c}

The choice of $c$ in \eqref{eq:rank}
heavily affects the performance of policies satisfying guardrails.
\begin{itemize}
\item
  If $c$ is too large,
  then guardrails may not differentiate between jobs of different sizes.
\item
  If $c$ is too small,
  then guardrails may misguidedly differentiate between jobs of similar sizes.
  This could allow one server to receive multiple small jobs of different ranks
  while another receives none.
\end{itemize}
To balance this tradeoff, we set $c$ to be a function of load~$\rho$:
\begin{equation}
  \label{eqn:def-c}
  c = 1 + \frac{1}{1 + \ln \frac{1}{1-\rho}}.
\end{equation}
This particular value of~$c$ is chosen to enable
the heavy-traffic optimality proof
for any dispatching policy satisfying guardrails.

\subsection{Guarded Policies:
  How to Augment Dispatching Policies with Guardrails}

Guardrails as described in Definition~\ref{def:guardrails}
are a set of constraints on dispatching policies
that we will use to guarantee bounds on mean response time
(see Section~\ref{sec:theorems}).
However, the constraints alone do not give a complete dispatching policy.
To define a concrete dispatching policy satisfying guardrails,
we start with an arbitrary dispatching policy~\P/
and augment it to create a new policy, called \emph{Guarded-\P/} (\Gp/),
which satisfies guardrails.

Roughly speaking, \Gp/ tries to dispatch according to~\P/,
but if dispatching to \P/'s favorite server would violate guardrails,
\Gp/ considers \P/'s second-favorite server, and so on.
Below are guarded versions of some common dispatching policies:
\begin{itemize}
\item
  G-Random dispatches to a random server
  among those which satisfy guardrails.
\item
  G-LWL dispatches to the server with the least remaining work
  among those which satisfy guardrails.
\item
  Round-Robin (RR) can be seen as
  always dispatching to the server that has least recently
  received a job,
  so G-RR dispatches to the server that has least recently received a job
  among those which satisfy guardrails.
\end{itemize}

Given an arbitrary dispatching policy~\P/,
Algorithm~\ref{alg:guarded} formally defines \Gp/.
We assume that P is specified by procedure
\Call{Dispatch\textsuperscript{\P/}}{}
which, when passed a job of size~$x$ and a set of servers~$\mathcal{S}$,
returns a server in $\mathcal{S}$ to which P would dispatch
a job of size~$x$.
The key to Algorithm~\ref{alg:guarded}
is that instead of calling \Call{Dispatch\textsuperscript{\P/}}{}
with the set of all servers,
we pass it a restricted set of servers
$\mathcal{S}_{\mathrm{safe}} \subseteq \mathcal{S}$
such that dispatching to any server in $\mathcal{S}_{\mathrm{safe}}$
will satisfy guardrails.
$\mathcal{S}_{\mathrm{safe}}$ is never empty because $x \leq gc^{r + 1}$,
so $\mathcal{S}_{\mathrm{safe}}$ will always contain the server~$s'$
of minimal $G^r_{s'}$.

\begin{algorithm}[t]
  \caption{Guarded-P (\Gp/)}
  \label{alg:guarded}

  \renewcommand{\algorithmicrequire}{\textit{Given:}}
  \newcommand{\Given}{\Require}
  \renewcommand{\algorithmicensure}{\textit{Call when:}}
  \newcommand{\CallWhen}{\smallskip\Ensure}

  \newlength{\alignwidth}
  \newcommand{\Gets}[2]{\State \makebox[\alignwidth][l]{#1}$\,{}\gets{}\,$#2}

  \begin{algorithmic}
    \Given
      dispatching policy~\P/,
      tightness $g \geq 1$,
      set of servers~$\mathcal{S}$,
      and rank width $c = 1 + \frac{1}{1 + \ln \frac{1}{1-\rho}}$

    \CallWhen the system starts
    \Procedure{Initialize\textsuperscript{\Gp/}}{{}}
      \For{each rank~$r$ and each server~$s \in \mathcal{S}$}
        \settowidth{\alignwidth}{$G^r_s$}
        \Gets{$G^r_s$}{$0$}
      \EndFor
    \EndProcedure

    \CallWhen a job of size~$x$ arrives
    \Procedure{Dispatch\textsuperscript{\Gp/}}{$x$}
      \settowidth{\alignwidth}{$\mathcal{S}_{\mathrm{safe}}$}
      \Gets{$r$}{$\floor{\log_c x}$}
      \Gets{$G_{\min}$}{$\min_{s' \in \mathcal{S}} G^r_{s'}$}
      \Gets{$\mathcal{S}_{\mathrm{safe}}$}{%
        $\{s' \in \mathcal{S} \mid G^r_{s'} + x \leq G_{\min} + gc^{r + 1}\}$}
      \Gets{$s$}{$\Call{Dispatch\textsuperscript{\P/}}{x, \mathcal{S}_{\mathrm{safe}}}$}
      \Gets{$G^r_s$}{$G^r_s + x$}
      \State \Return $s$
    \EndProcedure

    \CallWhen server~$s$ becomes empty
    \Procedure{Reset\textsuperscript{\Gp/}}{$s$}
      \For{each rank~$r$}
        \settowidth{\alignwidth}{$G^r_s$}
        \Gets{$G^r_s$}{$\min_{s' \in \mathcal{S}} G^r_{s'}$}
      \EndFor
    \EndProcedure
  \end{algorithmic}
\end{algorithm}

Algorithm~\ref{alg:guarded} is phrased in terms of
for loops over all ranks~$r$.
While there are infinitely many ranks in theory,
it is simple to represent all of the work counters in finite space
by representing most of them implicitly.

\subsection{Resets}
\label{sec:resets}

Algorithm~\ref{alg:guarded} includes a procedure,
\Call{Reset\textsuperscript{\Gp/}}{},
which we have not yet explained.
As defined so far,
guardrails effectively spread out small jobs across the servers,
but they have an unfortunate side effect:
they sometimes prevent dispatches to empty servers.
This is because the work counters~$G^r_s$ as defined so far
depend only on the dispatching history,
not the current server state.

Because dispatching to empty servers is desirable,
we would like to ensure that dispatching to an empty server
never violates guardrails.
We accomplish this by having servers \emph{reset} whenever they becomes empty.
When a server~$s$ resets, for each rank~$r$,
we decrease $G^r_s$ to match the minimum among all rank~$r$ work counters.
Because all rank~$r$ jobs have size less than~$gc^{r+1}$,
by Definition~\ref{def:guardrails},
dispatching to a server that has just reset will never violate guardrails.

\section{Technical Summary}
\label{sec:technical-summary}

\subsection{System Model}

We will study a $k$-server load balancing system
with Poisson arrivals at rate $\lambda$ jobs per second
and job size distribution~$X$.
Our optimality results (Theorem~\ref{thm:guard-srpt})
assume that either
\begin{itemize}
\item $X$ has bounded maximum size or
\item the tail of $X$ has upper Matuszewska index\matuszewska{} less than~$-2$.
\end{itemize}
This disjunctive assumption is roughly equivalent to finite variance.
We adopt the convention that each of the $k$ servers serves jobs at speed $1/k$.
As a result, a job of size $x$ requires $kx$ service time to complete.
We have chosen to define the speed of a server this way because
we will later compare the $k$-server system with a single server system of speed 1,
and this convention allows us to directly apply standard results
on single server systems.
We define the system load $\rho$ for both a single-server system
and the $k$-server system by
\[ \rho = \lambda \E{X} < 1. \]
Load does not depend on $k$ because the total service rate of all
$k$ servers combined is 1.

Throughout, we assume that the dispatching policy is a guarded policy,
as defined in Algorithm~\ref{alg:guarded}.
We consider two different scheduling policies that might be used at the servers:
\begin{description}
\item[\Srpt/]
  The policy that serves the job of least remaining size.
\item[\Prio/]
  The preemptive class-based priority policy
  in which a job's class is its \emph{rank}, as defined by \eqref{eq:rank}.
  That is, a job of size~$x$ has rank $r = \floor{\log_c x}$,
  and \Prio/ serves the job of minimal rank.
  Within each rank, jobs are served FCFS.
\end{description}

\subsection{Theorem Overview}

Our overall goal is to prove, for any dispatching policy~\P/,
the asymptotic optimality of the policy Guarded-P (\Gp/)
with respect to mean response time, given \Srpt/ scheduling at the servers.
We refer to this joint dispatch/scheduling policy as \GpSrpt/.

Rather than studying \GpSrpt/ directly,
we instead bound mean response time under \GpPrio/.
By the optimality of \Srpt/ scheduling~\cite{schrage-srpt-optimal},
the mean response time under \GpPrio/
gives an upper bound on the mean response time under \GpSrpt/.

\begin{theorem}
  \label{thm:guard-bound}
  For any dispatching policy~\P/,
  consider the policy \Gp/ with tightness~$g$.
  The expected response time for a job of size~$x$ under \GpPrio/
  is bounded by
  \begin{equation*}
    \E{T(x)}^{\GpPrio/}
    \le \frac{\frac{\lambda}{2} \int^{c^{r+1}}_0 t^2f_X(t)dt}{(1-\rho_{c^r})(1-\rho_{c^{r+1}})} + \frac{(4c+2)gk\frac{c^{r+1}}{c-1} + kx}{(1-\rho_{c^r})},
  \end{equation*}
  where
  \begin{itemize}
  \item
    $f_X(\cdot)$ is the probability density function of $X$,
  \item
    $c$ is the guardrail rank width
  \item
    $r = \floor{\log_c x}$ is the rank of a job of size~$x$, and
  \item
    $\rho_y = \lambda \int_0^y t f_X(t)\,dt$ is the load due to jobs of size $\le y$.
  \end{itemize}
\end{theorem}

We prove Theorem~\ref{thm:guard-bound} in Section~\ref{sec:guard-bound}.

Using the bound in Theorem~\ref{thm:guard-bound},
we show that the mean response time of the \GpPrio/ system
converges to that of a single-server \Srpt/ system
in the heavy-traffic limit.

\begin{theorem}
  \label{thm:guard-srpt}
  Consider a single-server \Srpt/ system whose single server is
  $k$ times as fast as each server in the load balancing system.
  For any dispatching policy~\P/,
  consider the policy \Gp/ with any constant tightness.
  Then for any size distribution~$X$ which is either
  (i)~bounded or
  (ii)~unbounded with tail having upper Matuszewska index\matuszewska{} less than~$-2$,
  the mean response times of \GpSrpt/, \GpPrio/, and (single-server) \Srpt/ converge
  as load approaches capacity:
  \[\lim_{\rho \to 1} \frac{\E{T}^{\GpSrpt/}}{\E{T}^{\Srpt/}} = \frac{\E{T}^{\GpPrio/}}{\E{T}^{\Srpt/}} = 1.\]
\end{theorem}

We prove Theorem~\ref{thm:guard-srpt} in Section~\ref{sec:guard-srpt}.

Theorem~\ref{thm:guard-srpt} relates the mean response times
of \GpSrpt/ and single-server \Srpt/,
which has the optimal mean response time
among all single-server policies~\cite{schrage-srpt-optimal}.
But a single-server system can simulate a load balancing system
running any joint dispatching/scheduling policy~\ArbArb/.
As a result, the mean response time under single-server \Srpt/
is a lower bound for the mean response time under \ArbArb/.

Using that bound,
Theorem~\ref{thm:guard-srpt} implies the following relationship between
the mean response times of \GpSrpt/ and~\ArbArb/.

\begin{corollary}
  \label{cor:guard-opt}
    For any dispatching policy~\P/
    consider the policy \Gp/ with any constant tightness.
    Consider any joint dispatching/scheduling policy \ArbArb/.
    Then for any size distribution~$X$ which is either
    (i)~bounded or
    (ii)~unbounded with tail having upper Matuszewska index\matuszewska{} less than~$-2$,
    the mean response times of \GpSrpt/ and \GpPrio/
    are at least as small as the mean response time of \ArbArb/
    as load approaches capacity:
    \[\lim_{\rho \to 1} \frac{\E{T}^{\GpSrpt/}}{\E{T}^{\ArbArb/}} = \frac{\E{T}^{\GpPrio/}}{\E{T}^{\ArbArb/}} \le 1.\]
\end{corollary}

\subsection{Relationship to Prior Work}

Our guardrails provide the first mechanism
to augment an arbitrary dispatching policy
to ensure size balance at all job size scales.
Moreover, we give the first bound
on the mean response time of load balancing systems
with \Srpt/ scheduling at the servers.
Using this bound, we prove that guarded dispatching policies
have asymptotically optimal mean response time
in the $\rho \to 1$ limit for any constant number of servers.
Our guarded policies are the first dispatching policies
known to have this property for general job size distributions.

We are not the first to consider load balancing systems
with \Srpt/ scheduling at the servers.
Avrahami and Azar~\cite{Avrahami:2003:MTF:777412.777415} consider
a problem analogous to ours but in a \emph{worst-case} setting,
assuming adversarial arrival times and job sizes,
in contrast to our stochastic setting.
Their dispatching policy, which they call IMD,
divides jobs into size ranks in a manner similar to our size ranks,
except with the width of each rank set to $c = 2$.
IMD dispatches each rank $r$ job to the server that has received the
least work of rank $r$ jobs in the past.
Put another way, IMD is the policy that keeps maximally tight guardrails
with no other underlying policy.
Avrahami and Azar prove that IMD is $O(\log P)$ competitive with
an optimal migratory offline algorithm,
where $P$ is the ratio between the largest and smallest job sizes in this system.
Note that $P$ can be arbitrarily large for general job size distributions.
Unfortunately,
the $O(\log P)$ competitive ratio is optimal for any online dispatching policy
in the worst-case setting~\cite{LEONARDI2007875}.
Our result is much stronger thanks to our stochastic setting.

Down and Wu~\cite{Down2006} also consider
a stochastic setting with \Srpt/ scheduling at the servers,
and they also propose a dispatching policy
that balances jobs of different sizes across the servers.
Their analysis does not result in any formula for mean response time
but instead uses a diffusion limit argument to show optimality in heavy traffic.
However, this limits their results to discrete job size distributions,
thus exlcuding many practically important continuous job size distributions.
In fact, Down and Wu~\cite[Section~5]{Down2006} observe empirically that
their policy performs poorly for Bounded Pareto job size distributions.
In contrast, our analysis shows that any guarded dispatching policy
is heavy-traffic optimal for general job size distributions,
including Bounded Pareto (see Figure~\ref{fig:bp-many-pol-g}).
Finally, the Down and Wu~\cite{Down2006} result
provides no insight into mean response time
outside of the heavy traffic regime,
whereas we derive a mean response time bound that is valid for all loads.

\section{Analysis of Guarded Policies}
\label{sec:theorems}

In this section, we analytically bound
the mean response time of a load balancing system
using an arbitrary guarded dispatching policy \Gp/ paired with \Srpt/ scheduling.
We then show that our bound implies that \GpSrpt/ minimizes mean response time
in heavy traffic.

\subsection{Preliminaries and Notation}
\label{sec:preliminaries}

We use a tagged job analysis:
we analyze the expected response time of a particular ``tagged'' job,
which we call~$j$, arriving to a steady-state system.
The expected response time of~$j$ is equal to the system's mean response time
by the PASTA property~\cite{wolff-pasta}.

Instead of studying \GpSrpt/ directly,
we analyze \GpPrio/,
which yields an upper bound on the mean response time under \GpSrpt/.
Studying \Prio/ simplifies the analysis because
the priority classes of \Prio/ match the ranks used by guardrails.

Suppose that $j$ has rank~$r$ and is dispatched to server~$s$.
Under \Prio/ scheduling,
there are two types of work that might delay job~$j$:
\begin{itemize}
\item
  The \emph{current relevant work} at server~$s$ when $j$ arrives.
  This is the total amount of remaining work at server~$s$
  due to jobs of rank~$\le r$.
\item
  The \emph{future relevant work} due to
  arriving jobs dispatched to server~$s$ while $j$ is in the system.
  These are the jobs dispatched to~$s$ of rank~$< r$
  (that is, rank~$\le r - 1$).
\end{itemize}

We use the following notation,
where ``rank~$r$ work'' denotes work due to jobs of rank~$r$.
\begin{itemize}
\item $W_s^r(t)$ denotes the current amount of rank~$r$ work
  at server~$s$ at time~$t$.
\item $V_s^r(t)$ denotes the total amount of rank~$r$ work
  that has ever been dispatched to server~$s$ up to time~$t$.
  In particular, the amount of rank~$r$ work dispatched to~$s$
  during a time interval $(t_1, t_2)$ is $V_s^r(t_2) - V_s^r(t_1)$.
\item $G_s^r(t)$ denotes the rank~$r$ guardrail work counter
  for server~$s$ at time~$t$ (see Algorithm~\ref{alg:guarded}).
  Specifically, $G_s^r(t)$ is defined as follows:\footnote{%
    The notations $t^-$ and $t^+$ below
    refer to ``just before'' and ``just after'' time~$t$.
    More formally, they refer to the left and right limits, respectively,
    of an expression that is piecewise-continuous in~$t$.}
  \begin{itemize}
  \item If a rank $r$ job of size $x$ is dispatched to server $s$ at time~$t$,
    we set $G_s^r(t^+) = G_s^r(t^-) + x$.
  \item If a server $s$ becomes empty of all jobs at time~$t$,
    we set $G_s^r(t^+) = \min_{s'} G_{s'}^r(t^-)$,
    where $s'$ ranges over all servers.
    We call this a \emph{reset} of server~$s$.
  \item Otherwise, $G_s^r(t)$ does not change.
  \end{itemize}
\end{itemize}
We write $W_s^{\le r}(t)$, $V_s^{\le r}(t)$, and $G_s^{\le r}(t)$
to denote the corresponding quantities where we consider all ranks~$\le r$,
rather than just rank~$r$,
and similarly for superscript~$< r$.

Occasionally, we will be talking about the \emph{total} work in the system,
or the \emph{total} work that has arrived, summed over all servers.
In that case, we will drop the subscript~$s$, writing $W^{\le r}(t)$ or $V^{\le r}(t)$.
Finally, we write $W^{\le r}$ to denote the stationary distribution of
the amount of rank~$\le r$ work in the whole system.

\subsection{Bounding Response Time: Key Steps}

Our goal in this section is to bound
the expected response time of a tagged job~$j$ under \GpPrio/.
We assume that $j$ has size~$x$ and rank $r = \floor{\log_cx}$.
We first bound current relevant work,
then move on to bound future relevant work.

We begin by showing that guardrails ensure that
any two servers have a similar amount of remaining rank~$\le r$ work.

\begin{lemma}
  \label{lem:work-balance}
  For any dispatching policy~\P/,
  consider the dispatching policy \Gp/ with tightness~$g$.
  In a \GpPrio/ system,
  the difference in remaining rank~$\le r$ work
  between any two servers $s$ and~$s'$ at any time~$t$
  is bounded by
  \begin{equation*}
    W_s^{\le r}(t) - W_{s'}^{\le r}(t) \le \frac{2gc^{r+2}}{c-1},
  \end{equation*}
  where $c$ is the guardrail rank width.
\end{lemma}

We prove Lemma~\ref{lem:work-balance} in Section~\ref{sec:work-balance}.

Roughly speaking,
Lemma~\ref{lem:work-balance} shows that guarded policies do a good job
of spreading out rank~$\le r$ work across the servers.
This is important because if the rank~$\le r$ work is spread out well,
then whenever there is a large amount of rank~$\le r$ work in the system,
\emph{all the servers} are doing rank~$\le r$ work.
This allows us to bound the amount of rank~$\le r$ work
in the $k$-server \GpPrio/ system
in terms of the remaining rank~$\le r$ work in an \MgPrio/ system
with a single server that runs $k$ times as fast.

\begin{lemma}
  \label{lem:rel-work-bound}
  For any dispatching policy~\P/,
  consider the dispatching policy \Gp/ with tightness~$g$.
  The total amount of remaining rank~$\le r$ work
  in a \GpPrio/ system
  is stochastically bounded relative to the remaining rank~$\le r$ work
  in a \MgPrio/ system
  whose server runs $k$ times as fast:
  \[W^{\le r} \le_\st W^{\le r}_{\MgPrio/} + \frac{2gkc^{r+2}}{c-1},\]
  where $c$ is the guardrail rank width.
\end{lemma}

We prove Lemma~\ref{lem:rel-work-bound} in Section~\ref{sec:rel-work-bound}.

Combining Lemmas~\ref{lem:work-balance} and~\ref{lem:rel-work-bound}
yields a bound on the amount of remaining rank~$\le r$ work
at any server~$s$,
thus bounding the current relevant work.

We now turn to bounding future relevant work.
Suppose that the tagged job~$j$ is dispatched to server~$s$.
The fact that guardrails spread out relevant work across the servers
means that while $j$ is in the system
$s$ will not receive much more rank~$< r$ work than other servers,
thus bounding future relevant work.
Combining this with our bound on current relevant work
yields the following bound on $j$'s response time.

\begin{lemma}
  \label{lem:first-bound}
  For any dispatching policy~\P/,
  consider the dispatching policy \Gp/ with tightness~$g$.
  In a \GpPrio/ system,
  the response time of a job of size~$x$
  is stochastically bounded by
  \[T(x) \le_\st B_{< r}\biggl(W^{\le r}_{\MgPrio/} + \frac{(4c+2)gkc^{r+1}}{c-1} + kx\biggr),\]
  where $c$ is the guardrail rank width,
  $r = \lfloor \log_c x \rfloor$ is the rank of the job,
  and $B_{<r}(w)$ is the length of a busy period comprising
  only jobs of rank~$< r$ started by work~$w$.
\end{lemma}

We prove Lemma~\ref{lem:first-bound} in Section~\ref{sec:first-bound}.

Taking expectations in Lemma~\ref{lem:first-bound}
and applying the well-known formula for $\E{W^{\le r}_{\MgPrio/}}$,
we obtain Theorem~\ref{thm:guard-bound}.

\begin{reptheorem}{thm:guard-bound}
  For any dispatching policy~\P/,
  consider the policy \Gp/ with tightness~$g$.
  The expected response time for a job of size~$x$ under \GpPrio/
  is bounded by
  \begin{equation*}
    \E{T(x)}^{\GpPrio/}
    \le \frac{\frac{\lambda}{2} \int^{c^{r+1}}_0 t^2f_X(t)dt}{(1-\rho_{c^r})(1-\rho_{c^{r+1}})} + \frac{(4c+2)gk\frac{c^{r+1}}{c-1} + kx}{(1-\rho_{c^r})},
  \end{equation*}
  where
  \begin{itemize}
  \item
    $f_X(\cdot)$ is the probability density function of $X$,
  \item
    $c$ is the guardrail rank width
  \item
    $r = \floor{\log_c x}$ is the rank of a job of size~$x$, and
  \item
    $\rho_y = \lambda \int_0^y t f_X(t)\,dt$ is the load due to jobs of size $\le y$.
  \end{itemize}
\end{reptheorem}

We prove Theorem~\ref{thm:guard-bound} in Section~\ref{sec:guard-bound-proof}.

\subsection{Bounding Response Time: Proofs}
\label{sec:guard-bound}

\subsubsection{Proof of Lemma~\ref{lem:work-balance}}
\label{sec:work-balance}

\begin{replemma}{lem:work-balance}
  For any dispatching policy~\P/,
  consider the dispatching policy \Gp/ with tightness~$g$.
  In a \GpPrio/ system,
  the difference in remaining rank~$\le r$ work
  between any two servers $s$ and~$s'$ at any time~$t$
  is bounded by
  \begin{equation*}
    W_s^{\le r}(t) - W_{s'}^{\le r}(t) \le \frac{2gc^{r+2}}{c-1},
  \end{equation*}
  where $c$ is the guardrail rank width.
\end{replemma}

\begin{proof}
  Let $t_0$ be the most recent time up to time $t$ when server $s$ was empty of
  rank~$\le r$ work. Note that $t_0$ may equal $t$.
  We will bound the difference in rank $\le r$ work at the two servers at time $t$
  by comparison with time $t_0$.

  The remaining rank~$\le r$ work present at time $t$ is
  \begin{enumerate}[(i)]
  \item
    the remaining rank~$\le r$ work present at time $t_0$
  \item
    plus rank~$\le r$ work due to arrivals in the interval $[t_0, t]$
  \item
    minus rank $\le r$ work processed during the interval.
  \end{enumerate}
  We consider these three quantities first for server~$s$,
  then for server~$s'$.

  We begin with server~$s$:
  \begin{enumerate}[(i)]
  \item
    By the definition of $t_0$, there is no remaining rank~$\le r$ work at server $s$.
  \item
    The amount of work that arrives to server $s$ over the interval $[t_0, t]$
    is $V_s^{\le r}(t) - V_s^{\le r}(t_0)$.
  \item
    The amount of rank~$\le r$ work processed during the interval $[t_0, t]$
    is equal to $\frac{t-t_0}{k}$,
    because the server $s$ processes work at speed $1/k$,
    $s$ has rank~$\le r$ work available throughout the interval,
    and the \Prio/ scheduling policy always prioritizes lower rank work.
  \end{enumerate}
  These quantities give us
  the remaining rank~$\le r$ work at server $s$ at time $t$:
  \begin{equation}
    \label{eqn:s-1}
    W_s^{\le r}(t) = (V_s^{\le r}(t) - V_s^{\le r}(t_0)) - \frac{t - t_0}{k}.
  \end{equation}
  Because server $s$ was never empty at any time during $(t_0, t]$,
  the guardrail work counters were never reset to the system-wide minimums
  during $[t_0, t]$. As a result, the changes in $G_s^{\le r}(\cdot)$
  and $V_s^{\le r}(\cdot)$ over the interval $[t_0, t]$ must be equal.
  We can apply this fact to~(\ref{eqn:s-1}):
  \begin{equation}
    \label{eqn:s-2}
    W_s^{\le r}(t) = (G_s^{\le r}(t) - G_s^{\le r}(t_0)) - \frac{t - t_0}{k}.
  \end{equation}

  We now turn to server $s'$:
  \begin{enumerate}[(i)]
  \item
    The remaining rank~$\le r$ work present at server $s'$ at time $t_0$
    is non-negative.
  \item
    The amount of work that arrives to server $s$ over the interval $[t_0, t]$
    is $V_{s'}^{\le r}(t) - V_{s'}^{\le r}(t_0)$.
  \item
    The amount of rank $\le r$ work processed over the interval $[t_0, t]$ is at most
    $\frac{t-t_0}{k}$.
  \end{enumerate}
  Therefore, we may lower bound the remaining rank~$\le r$ work at server $s'$
  at time $t$:
  \begin{equation}
    \label{eqn:s'-1}
    W_{s'}^{\le r}(t) \ge (V_{s'}^{\le r}(t) - V_{s'}^{\le r}(t_0)) - \frac{t - t_0}{k}.
  \end{equation}
  The change in $G_{s'}^{\le r}(\cdot)$ over the interval $[t_0, t]$
  is no more than the change in $V_{s'}^{\le r}(\cdot)$
  over the same interval, since any reset to the system-wide minimum can
  only lead to a decrease in $G_{s'}^{\le r}(t)$.
  We can apply this fact to~(\ref{eqn:s'-1}):
  \begin{equation}
    \label{eqn:s'-2}
    W_{s'}^{\le r}(t) \ge (G_{s'}^{\le r}(t) - G_{s'}^{\le r}(t_0)) - \frac{t - t_0}{k}.
  \end{equation}

  Combining (\ref{eqn:s-2}), (\ref{eqn:s'-2}), and the guardrail constraint
  in Definition~\ref{def:guardrails}
  yields the desired bound:
  \begin{align*}
    \ifwidecol{}{\MoveEqLeft} W_s^{\le r}(t) - W_{s'}^{\le r}(t) \ifwidecol{}{\\}
    &\le (G_s^{\le r}(t) - G_s^{\le r}(t_0)) - (G_{s'}^{\le r}(t) - G_{s'}^{\le r}(t_0))\\
    &\le |G_s^{\le r}(t) - G_{s'}^{\le r}(t)| + |G_{s}^{\le r}(t_0) - G_{s'}^{\le r}(t_0)|\\
    &= \sum_{q = -\infty}^r (|G_s^q(t) - G_{s'}^q(t)| + |G_{s'}^q(t_0) - G_s^q(t_0)|)\\
    &\le \sum\limits_{q = -\infty}^r 2gc^{q+1}\\
    &= \frac{2gc^{r+2}}{c-1}.
    \qedhere
  \end{align*}
\end{proof}

\subsubsection{Proof of Lemma~\ref{lem:rel-work-bound}}
\label{sec:rel-work-bound}

\begin{replemma}{lem:rel-work-bound}
  For any dispatching policy~\P/,
  consider the dispatching policy \Gp/ with tightness~$g$.
  The total amount of remaining rank~$\le r$ work
  in a \GpPrio/ system
  is stochastically bounded relative to the remaining rank~$\le r$ work
  in a \MgPrio/ system
  whose server runs $k$ times as fast:
  \[W^{\le r} \le_\st W^{\le r}_{\MgPrio/} + \frac{2gkc^{r+2}}{c-1},\]
  where $c$ is the guardrail rank width.
\end{replemma}

\begin{proof}
We consider two coupled systems receiving the same arrivals:
\begin{itemize}
\item
  a \GpPrio/ system where each of the $k$ servers runs at speed $1/k$, and
\item
  a \MgPrio/ system where the single server runs at speed $1$.
\end{itemize}
We will refer to the total amount of remaining rank~$\le r$ work in the
\GpPrio/ system as $W^{\le r}(t)$,
and the total amount of remaining rank~$\le r$ work
in the \MgPrio/ system as $W^{\le r}_{\MgPrio/}(t)$.

It suffices to show that at any time $t$,
we have the following bound on
the difference in the total amounts of remaining rank~$\le r$ work
between the two systems:
\begin{equation}
  \label{eqn:couple}
  W^{\le r}(t) \le W^{\le r}_{\MgPrio/}(t) + \frac{2gkc^{r+2}}{c-1}.
\end{equation}
To prove (\ref{eqn:couple}), we consider two cases:
\begin{enumerate}[(i)]
\item
  At least one server in the \GpPrio/ system
  that has no remaining rank~$\le r$ work at time~$t$.
\item
  All servers in the \GpPrio/ system
  have remaining rank~$\le r$ work at time~$t$.
\end{enumerate}

In case~(i), suppose server $s'$ in the \GpPrio/ system
has no remaining rank~$\le r$ work at time $t$.
By Lemma~\ref{lem:work-balance}, we know that at all servers $s$,
\[W^{\le r}_s(t) = W^{\le r}_s(t) - W^{\le r}_{s'}(t) \le \frac{2gc^{r+2}}{c-1}.\]
Summing over all $k$ servers implies~(\ref{eqn:couple}).

We now turn to case~(ii).
Let $t_0$ be the most recent time before $t$ when a \GpPrio/ server
had no remaining rank~$\le r$ work.
Note that case~(i) applies at time~$t_0$.
Therefore,
it suffices to show that the difference in remaining rank~$\le r$ work
between the two systems is no more at time~$t$ than at time~$t_0$:
\begin{equation}
  \label{eqn:many}
  W^{\le r}(t) - W^{\le r}_{\MgPrio/}(t) \le W^{\le r}(t_0) - W^{\le r}_{\MgPrio/}(t_0).
\end{equation}
By definition of~$t_0$,
for the duration of entire time interval $(t_0, t)$,
each of the $k$~servers in the \GpPrio/ system
processes rank~$\le r$ work at speed $1/k$, for a total of $t - t_0$ work.
This is at least as much rank~$\le r$ work as
the \MgPrio/ system processes during $(t_0, t)$,
because the single server's speed is~$1$.
Due to coupling, the two systems receive the same amount of
rank~$\le r$ work during $(t_0, t)$,
implying~\eqref{eqn:many}.
\end{proof}

\subsubsection{Proof of Lemma~\ref{lem:first-bound}}
\label{sec:first-bound}

\begin{replemma}{lem:first-bound}
  For any dispatching policy~\P/,
  consider the dispatching policy \Gp/ with tightness~$g$.
  In a \GpPrio/ system,
  the response time of a job of size~$x$
  is stochastically bounded by
  \[T(x) \le_\st B_{< r}\biggl(W^{\le r}_{\MgPrio/} + \frac{(4c+2)gkc^{r+1}}{c-1} + kx\biggr),\]
  where $c$ is the guardrail rank width,
  $r = \lfloor \log_c x \rfloor$ is the rank of the job,
  and $B_{<r}(w)$ is the length of a busy period comprising
  only jobs of rank~$< r$ started by work~$w$.
\end{replemma}

\begin{proof}
Let
\begin{itemize}
\item
  $j$ be the tagged job,
\item
  $x$ be $j$'s size and $r = \floor{\log_c x}$ be $j$'s rank,
\item
  $a_j$ and $d_j$ be $j$'s arrival and departure times, respectively, and
\item
  $s$ be the server to which $j$ is dispatched.
\end{itemize}

The time at which job~$j$ departs, $d_j$,
can be calculated as the time
required for server~$s$ to complete the following work:
\begin{itemize}
\item
  relevant work already present at time~$a_j$,
  namely $W_s^{\le r}(a_j)$;
\item
  plus all relevant work that arrives at $s$ while $j$ is in the system,
  namely $V_s^{<r}(t) - V_s^{<r}(a_j)$;
\item
  plus $j$'s size, namely~$x$.
\end{itemize}
Let $t \geq a_j$ be an arbitrary time while~$j$ is in the system.
Because each server runs at speed $1/k$,
the amount of work that server~$s$ has completed by time $t$ is
$(t - a_j)/k$.
Writing
\begin{equation*}
  Z_s(t) = W_s^{\le r}(a_j) + (V_s^{<r}(t) - V_s^{<r}(a_j))
\end{equation*}
gives the following expression for $j$'s departure time~$d_j$:
\begin{equation}
  \label{eqn:djZ}
  d_j = \inf\biggl\{t \biggm| \frac{t - a_j}{k} \ge Z_s(t) + x\biggr\}.
\end{equation}

To bound $d_j$, we first bound $Z_s(t)$.
Let
\begin{equation}
  \label{eq:z-bar}
  \bar{Z}(t) = \frac{1}{k}\sum_{s'}Z_{s'}(t)
\end{equation}
be the average value of $Z_{s'}(t)$ over all servers~$s'$,
and let
\begin{equation*}
  Z_s^{\mathrm{maxdiff}}(t) = \max_{s'} (Z_s(t) - Z_{s'}(t))
\end{equation*}
be the maximum difference between $Z_s(t)$ and $Z_{s'}(t)$ over all servers~$s'$.
Observe that
\begin{equation*}
  Z_s(t) \le \bar{Z}(t) + Z_s^{\mathrm{maxdiff}}(t).
\end{equation*}
Combining this with (\ref{eqn:djZ}) gives a bound on $d_j$:
\begin{equation}
  \label{eqn:end-bound}
  d_j \le \inf \biggl\{t \biggm| \frac{t - a_j}{k} \ge
  \bar{Z}(t) + Z_s^{\mathrm{maxdiff}}(t) + x\biggr\}.
\end{equation}

To simplify (\ref{eqn:end-bound}), our next step is to bound $Z_s^{\mathrm{maxdiff}}(t)$.
We start by expanding $Z_s(t) - Z_{s'}(t)$:
\ifwidecol{\begin{equation}}{\begin{align}}
  \ifwidecolignorealigny
  Z_s(t) - Z_{s'}(t)
  \ifwidecol{}{&}= (W_s^{\le r}(a_j) - W_{s'}^{\le r}(a_j)) \\
  \notag
  \ifwidecol{}{&\qquad} + (V_s^{< r}(t) - V_s^{<r}(a_j)) \\
  \label{eq:z-expanded}
  \ifwidecol{}{&\qquad} - (V_{s'}^{< r}(t) - V_{s'}^{<r}(a_j)).
\ifwidecol{\end{equation}}{\end{align}}
We are left with an expression in terms of
the rank~$< r$ work dispatched to each server.
We would like to turn this into an expression in terms of
guardrail work counters,
which will allow us to apply the constraints
given by Definition~\ref{def:guardrails}.
Consider the time interval $(a_j, t)$.
Job~$j$ is present at server $s$ for the duration of the interval,
so server~$s$ does not reset,
implying
\begin{equation}
  \label{eq:v-vs-g-s}
  V_s^{< r}(t) - V_s^{< r}(a_j) = G_s^{< r}(t) - G_s^{< r}(a_j).
\end{equation}
In contrast, server~$s'$ may reset during $(a_j, t)$.
When a reset occurs at some time~$t_{\text{reset}}$,
$G_{s'}^{< r}(t_{\text{reset}})$ decreases while $V_{s'}^{< r}(t_{\text{reset}})$ stays constant.
Furthermore, $G_{s'}^{< r}(t')$ and $V_{s'}^{< r}(t')$ change in the same way
at all other times~$t'$, so
\begin{equation}
  \label{eq:v-vs-g-s'}
  V_{s'}^{< r}(t) - V_{s'}^{< r}(a_j) \geq G_{s'}^{< r}(t) - G_{s'}^{< r}(a_j).
\end{equation}
Applying \eqref{eq:v-vs-g-s}, \eqref{eq:v-vs-g-s'},
and Lemma~\ref{lem:work-balance} to \eqref{eq:z-expanded} yields the bound
\ifwidecol{\begin{equation*}}{\begin{align*}}
  \ifwidecolignorealigny
  \MoveEqLeft Z_s(t) - Z_{s'}(t)\\
  \ifwidecol{}{&}
  \le \frac{2gc^{r+2}}{c-1} + (G_s^{< r}(t) - G_{s}^{< r}(a_j)) - (G_{s'}^{<r}(t) - G_{s'}^{<r}(a_j)).
\ifwidecol{\end{equation*}}{\end{align*}}

Because \Gp/ is a guarded policy,
we can apply the guardrail constraints from Definition~\ref{def:guardrails}:
\begin{align*}
  \ifwidecol{}{\MoveEqLeft} Z_s(t) - Z_{s'}(t) \ifwidecol{}{\\}
  &\leq \frac{2gc^{r+2}}{c-1} + (G_s^{< r}(t) - G_{s'}^{< r}(t)) - (G_{s}^{<r}(a_j) - G_{s'}^{<r}(a_j))\\
  &= \frac{2gc^{r+2}}{c-1} + \sum_{q=-\infty}^{r-1} (G_s^q(t) - G_{s'}^q(t))
    + (G_{s'}^q(a_j) - G_s^q(a_j))\\
  &\leq \frac{2gc^{r+2}}{c-1} + \sum_{q=-\infty}^{r-1} (gc^{q+1} + gc^{q+1}) \\
  &= \frac{(2c + 2)gc^{r+1}}{c-1}.
\end{align*}
We have bounded $Z_s(t) - Z_{s'}(t)$ for arbitrary~$s'$
and hence bounded $Z_s^{\mathrm{maxdiff}}(t)$.
Substituting into (\ref{eqn:end-bound}) yields
\begin{equation*}
  d_j \le \inf \biggl\{t \biggm| \frac{t - a_j}{k} \ge
  \bar{Z}(t) + \frac{(2c + 2)gc^{r+1}}{c-1} + x\biggr\}.
\end{equation*}
Recalling the definition of $\bar{Z}(t)$ from \eqref{eq:z-bar} gives us
\ifwidecol{\begin{equation*}}{\begin{align*}}
  \ifwidecolignorealigny
  d_j \le \inf\biggl\{t \biggm| t - a_j
  \ifwidecol{}{&}\ge W^{\le r}(a_j) + (V^{<r}(t) - V^{<r}(a_j))\\
  \ifwidecol{}{&\qquad} + \frac{(2c+2)gkc^{r+1}}{c-1} + kx\biggr\}.
\ifwidecol{\end{equation*}}{\end{align*}}
Because the arrival process to the overall system is a Poisson process,
we can rewrite this in terms of a ``relevant'' busy period,
meaning one containing only jobs of rank~$< r$:
\begin{equation*}
  d_j - a_j \le B_{<r}\biggl(W^{\le r}(a_j) + \frac{(2c+2)gkc^{r+1}}{c-1} + kx\biggr).
\end{equation*}
The Poisson arrival process also implies,
by the PASTA property~\cite{wolff-pasta},
that the amount of relevant work $j$ sees on arrival, namely~$W^{\le r}(a_j)$,
is drawn from the steady-state distribution, namely~$W^{\le r}$, so
\begin{equation*}
  T(x) \le_\st B_{<r}\biggl(W^{\le r} + \frac{(2c+2)gkc^{r+1}}{c-1} + kx\biggr).
\end{equation*}
Applying Lemma~\ref{lem:rel-work-bound} to $W^{\le r}$ yields the desired bound.
\end{proof}

\begin{remark}
  \label{rmk:resets-optional}
  Note we can prove Lemmas~\ref{lem:work-balance} and~\ref{lem:first-bound}
  using only the following properties of resets:
  \begin{itemize}
  \item
    A server only resets when it is empty.
  \item
    When a server resets, its work counters do not increase.
  \item
    The guardrail constraints in Definition~\ref{def:guardrails}
    continue to hold after each reset.
  \end{itemize}
  In particular, this means that resets are optional
  for proving our response time bounds,
  so the bounds hold even if the dispatcher chooses to omit some resets.
  This is helpful when implementing guarded dispatching policies in large systems
  (see Sections~\ref{sec:network-delays} and~\ref{sec:multiple-dispatchers}).
\end{remark}

\subsubsection{Proof of Theorem~\ref{thm:guard-bound}}
\label{sec:guard-bound-proof}

\begin{reptheorem}{thm:guard-bound}
  For any dispatching policy~\P/,
  consider the policy \Gp/ with tightness~$g$.
  The expected response time for a job of size~$x$ under \GpPrio/
  is bounded by
  \begin{equation*}
    \E{T(x)}^{\GpPrio/}
    \le \frac{\frac{\lambda}{2} \int^{c^{r+1}}_0 t^2f_X(t)dt}{(1-\rho_{c^r})(1-\rho_{c^{r+1}})} + \frac{(4c+2)gk\frac{c^{r+1}}{c-1} + kx}{(1-\rho_{c^r})},
  \end{equation*}
  where
  \begin{itemize}
  \item
    $f_X(\cdot)$ is the probability density function of $X$,
  \item
    $c$ is the guardrail rank width
  \item
    $r = \floor{\log_c x}$ is the rank of a job of size~$x$, and
  \item
    $\rho_y = \lambda \int_0^y t f_X(t)\,dt$ is the load due to jobs of size $\le y$.
  \end{itemize}
\end{reptheorem}

\begin{proof}
  Recall the conclusion of Lemma~\ref{lem:first-bound},
  \begin{equation}
    \label{eqn:busy}
    T(x) \le B_{< r}\biggl(W^{\le r}_{\MgPrio/} + \frac{(4c+2)gkc^{r+1}}{c-1} + kx\biggr).
  \end{equation}
  Standard results on busy periods~\cite{mor-book} state that
  \[\E{B_{< r}(Y)} = \frac{\E{Y}}{1-\rho_{c^r}},\]
  and standard results on the single-server \Prio/ system\cite{mor-book}
  give the expected steady-state remaining rank~$\le r$ work:
  \[\E{W^{\le r}_{\MgPrio/}} = \frac{\frac{\lambda}{2} \int^{c^{r+1}}_0 t^2f_X(t)dt}{(1-\rho_{c^{r+1}})}.\]
  Taking expectations of (\ref{eqn:busy}) and applying these standard results
  yields the desired bound.
\end{proof}

\subsection{Asymptotic Behavior of Guarded Policies}
\label{sec:guard-srpt}
\begin{reptheorem}{thm:guard-srpt}
  Consider a single-server \Srpt/ system whose single server is
  $k$ times as fast as each server in the load balancing system.
  For any dispatching policy~\P/,
  consider the policy \Gp/ with any constant tightness.
  Then for any size distribution~$X$ which is either
  (i)~bounded or
  (ii)~unbounded with tail having upper Matuszewska index\matuszewska{} less than~$-2$,
  the mean response times of \GpSrpt/, \GpPrio/, and (single-server) \Srpt/ converge
  as load approaches capacity:
  \[\lim_{\rho \to 1} \frac{\E{T}^{\GpSrpt/}}{\E{T}^{\Srpt/}} = \frac{\E{T}^{\GpPrio/}}{\E{T}^{\Srpt/}} = 1.\]
\end{reptheorem}
\begin{proof}
  We start with the bound on the mean response time of
  \GpPrio/ from Theorem~\ref{thm:guard-bound}:
  \ifwidecol{\begin{equation}}{\begin{align}}
    \ifwidecolignorealigny
    \label{eqn:thm-bound}
    \MoveEqLeft \E{T(x)}^{\Gp//\Prio/} \nonumber \\
    \ifwidecol{}{&}\le \frac{\frac{\lambda}{2} \int^{c^{r+1}}_0 t^2f_X(t)dt}{(1-\rho_{c^r})(1-\rho_{c^{r+1}})} + \frac{(4c+2)gk\frac{c^{r+1}}{c-1} + kx}{(1-\rho_{c^r})}.
  \ifwidecol{\end{equation}}{\end{align}}
  Note that the first term in this expression also appears in the expression for
  mean response time under single-server \Prio/ \cite{mor-book}:
  \[\E{T(x)}^{\Prio/} = \frac{\frac{\lambda}{2} \int^{c^{r+1}}_0 t^2f_X(t)dt}{(1-\rho_{c^r})(1-\rho_{c^{r+1}})} + \frac{x}{(1-\rho_{c^r})}.\]
  Therefore, let us simplify~(\ref{eqn:thm-bound}):
  \ifwidecol{\begin{equation}}{\begin{align}}
    \ifwidecolignorealigny
    \label{eqn:prio-bound-1}
    \MoveEqLeft \E{T(x)}^{\Gp//\Prio/} \nonumber \\
    \ifwidecol{}{&}\le \E{T(x)}^{\Prio/} + \frac{(4c+2)gk\frac{c^{r+1}}{c-1} + (k-1)x}{(1-\rho_{c^r})}.
  \ifwidecol{\end{equation}}{\end{align}}
  We may simplify this bound by combining constant terms.
  Note that $c \le 2$ and $x \ge c^r$.
  Let $m=20gk + k - 1$.
  Then
  \[mx \ge (4c+2)gkc^{r+1} + (k - 1)(c-1)x.\]
  Thus, we may simplify (\ref{eqn:prio-bound-1}) further:
  \begin{equation*}
    \E{T(x)}^{\Gp//\Prio/}
    \le \E{T(x)}^{\Prio/} + \frac{mx}{(c-1)(1-\rho_{c^r})}.
  \end{equation*}

  We want to relate this to something more similar to the mean response time
  under \Srpt/. A convenient policy for comparison with \Prio/ that is
  similar enough to \Srpt/ is Preemptive-Shortest-Job-First (\Psjf/),
  which prioritizes jobs according to their original size.

    We now use Lemma~\ref{lem:prio-psjf} from Appendix~\ref{app:prio-psjf},
    which says that
    \[\E{T}^{\Prio/} \le (c+2\sqrt{c-1})\E{T}^{\Psjf/}.\]
    For brevity, let
    \[b(c) = (c+2\sqrt{c-1}).\]
    Using Lemma~\ref{lem:prio-psjf}, we find that
  \begin{equation*}
      \E{T(x)}^{\GpPrio/} \le b(c)\E{T}^{\Psjf/} + \frac{mx}{(c-1)(1-\rho_{c^r})}.
  \end{equation*}
  Note that $x \ge c^r$, so $\rho_x \ge \rho_{c^r}$, so
  \begin{equation*}
      \E{T(x)}^{\GpPrio/} \le b(c)\E{T}^{\Psjf/} + \frac{mx}{(c-1)(1-\rho_x)}.
  \end{equation*}
  Based on standard results on mean response time under \Psjf/ and \Srpt/~\cite{mor-book},
  we know that
  \[\E{T(x)}^{\Psjf/} \le \E{T(x)}^{\Srpt/} + \frac{x}{1-\rho_x}.\]
  Let $m' = m + 4$. Because $c \le 2$, we know $b(c)\cdot(c-1) \le 4$. Thus,
  \begin{equation}
    \label{eqn:srpt-bound}
      \E{T(x)}^{\GpPrio/} \le b(c)\E{T(x)}^{\Srpt/} + \frac{m' x}{(c-1)(1-\rho_x)}.
  \end{equation}

  Next, we take the expectation of (\ref{eqn:srpt-bound}) over all job sizes $x$.
  To do so, we need to integrate
  \[\int_0^\infty \frac{x}{1-\rho_x}f_X(x)dx,\]
  where $f_X(\cdot)$ is the probability density function of $X$.
  To compute the integral,
  we make a change of variables from $x$ to $\rho_x$,
  using the following facts:
  \begin{align*}
    \rho_x &= \int_0^x \lambda t f_X(t) dt\\
    \frac{d\rho_x}{dx} &= \lambda x f_X(x)\\
    \rho_0 &= 0\\
    \lim_{x \to \infty} \rho_x &= \rho.
  \end{align*}
  Given this change of variables, we compute
  \begin{align*}
    \int_0^\infty \frac{x}{1-\rho_x}f_X(x)dx
    = \int_0^\rho\frac{1}{\lambda(1-\rho_x)}d\rho_x
    =\frac{1}{\lambda}\ln \biggl(\frac{1}{1-\rho} \biggr).
  \end{align*}
  Applying this to~(\ref{eqn:srpt-bound}), we find that
  \begin{equation*}
      \E{T}^{\GpPrio/} \le b(c)\E{T}^{\Srpt/} + \frac{m'}{\lambda(c-1)}\ln \biggl(\frac{1}{1-\rho} \biggr).
  \end{equation*}

  Dividing through by $\E{T}^{\Srpt/}$, we find that
  \begin{equation*}
      \frac{\E{T}^{\GpPrio/}}{\E{T}^{\Srpt/}} \le b(c) + \frac{m'\ln\frac{1}{1-\rho}}{\lambda(c-1)\E{T}^{\Srpt/}}.
  \end{equation*}
  Plugging in the value of $c$ in terms of $\rho$ from~(\ref{eqn:def-c}),
  \begin{equation*}
    \frac{\E{T}^{\GpPrio/}}{\E{T}^{\Srpt/}}
    \le b(c) + \frac{m'\cdot(\ln^2\frac{1}{1-\rho} + \ln\frac{1}{1-\rho})}{\lambda\E{T}^{\Srpt/}}.
  \end{equation*}
  We now take the limit of the above ratio as  $\rho \to 1$.
  In this limit,
  \begin{itemize}
  \item $\ln\frac{1}{1-\rho}$ diverges,
  \item $\lambda$ approaches $\E{X}$, and
  \item $c$ approaches~$1$, so $b(c)$ also approaches~$1$:
    \begin{equation*}
      \lim_{c \to 1+} c + 2\sqrt{c-1} = 1.
    \end{equation*}
  \end{itemize}
  Therefore, letting $m'' = 2m'/\E{X}$,
  \begin{equation}
    \label{eqn:final}
    \lim_{\rho \to 1} \frac{\E{T}^{\Gp//\Prio/}}{\E{T}^{\Srpt/}}
    \le\lim_{\rho \to 1}\biggl( 1 + \frac{m'' \ln^2\frac{1}{1-\rho}}{\E{T}^{\Srpt/}}\biggr).
  \end{equation}

  Recall now that we assume that $X$ is either (i)~bounded or (ii)~unbounded with
  tail function having upper Matuszewska index\matuszewska{} less than -2.
  In Lemma~\ref{lem:lin-bound} in Appendix~\ref{app:lin-bound},
  we use a result of Lin et al.~\cite{Lin} to show that in either case,
  \begin{equation*}
    \lim_{\rho \to 1}\frac{\ln^2\frac{1}{1-\rho}}{\E{T}^{\Srpt/}} = 0.
  \end{equation*}
  Applying this to (\ref{eqn:final}), we find that
  \begin{equation}
    \label{eqn:last}
    \lim_{\rho \to 1} \frac{\E{T}^{\Gp//\Prio/}}{\E{T}^{\Srpt/}} \le 1.
  \end{equation}
  \Srpt/ yields optimal mean response time
  over all single-server policies~\cite{schrage-srpt-optimal},
  and a joint dispatching/scheduling policy
  can be emulated on a single server,
  so (\ref{eqn:last}) is in fact an equality, as desired.

  The optimality of \Srpt/'s mean response time also implies that
  the mean response time under \GpSrpt/ is no more
  than the mean response time under \GpPrio/.
  As a result,
  \[\lim_{\rho \to 1} \frac{\E{T}^{\GpSrpt/}}{\E{T}^{\Srpt/}} = 1. \qedhere\]
\end{proof}

\subsection{Optimality of Guarded Policies}
\label{sec:guard-opt}
As a simple corollary of Theorem~\ref{thm:guard-srpt},
we find that for any dispatching policy~\P/,
\Gp//\Srpt/
has optimal mean response time in the heavy traffic limit
over all joint dispatching/scheduling policies.
\begin{repcorollary}{cor:guard-opt}
    For any dispatching policy~\P/
    consider the policy \Gp/ with any constant tightness.
    Consider any joint dispatching/scheduling policy \ArbArb/.
    Then for any size distribution~$X$ which is either
  (i)~bounded or
  (ii)~unbounded with tail having upper Matuszewska index\matuszewska{} less than~$-2$,
    the mean response times of \GpSrpt/ and \GpPrio/
    are at least as small as the mean response time of \ArbArb/
    as load approaches capacity:
    \[\lim_{\rho \to 1} \frac{\E{T}^{\GpSrpt/}}{\E{T}^{\ArbArb/}} = \frac{\E{T}^{\GpPrio/}}{\E{T}^{\ArbArb/}} \le 1.\]
\end{repcorollary}
\begin{proof}
  \Srpt/ has optimal mean response time
  among all single-server policies~\cite{schrage-srpt-optimal},
  and any joint dispatching/scheduling policy
  can be emulated on a single server,
  so $\E{T}^{\Srpt/} \le \E{T}^{\ArbArb/}$.
  The result thus follows from Theorem~\ref{thm:guard-srpt}.
\end{proof}

\begin{figure*}[t]
  \begin{subfigure}{.48\linewidth}
    \centering
    \includegraphics[width=\linewidth, trim={0 1em 0 0}]{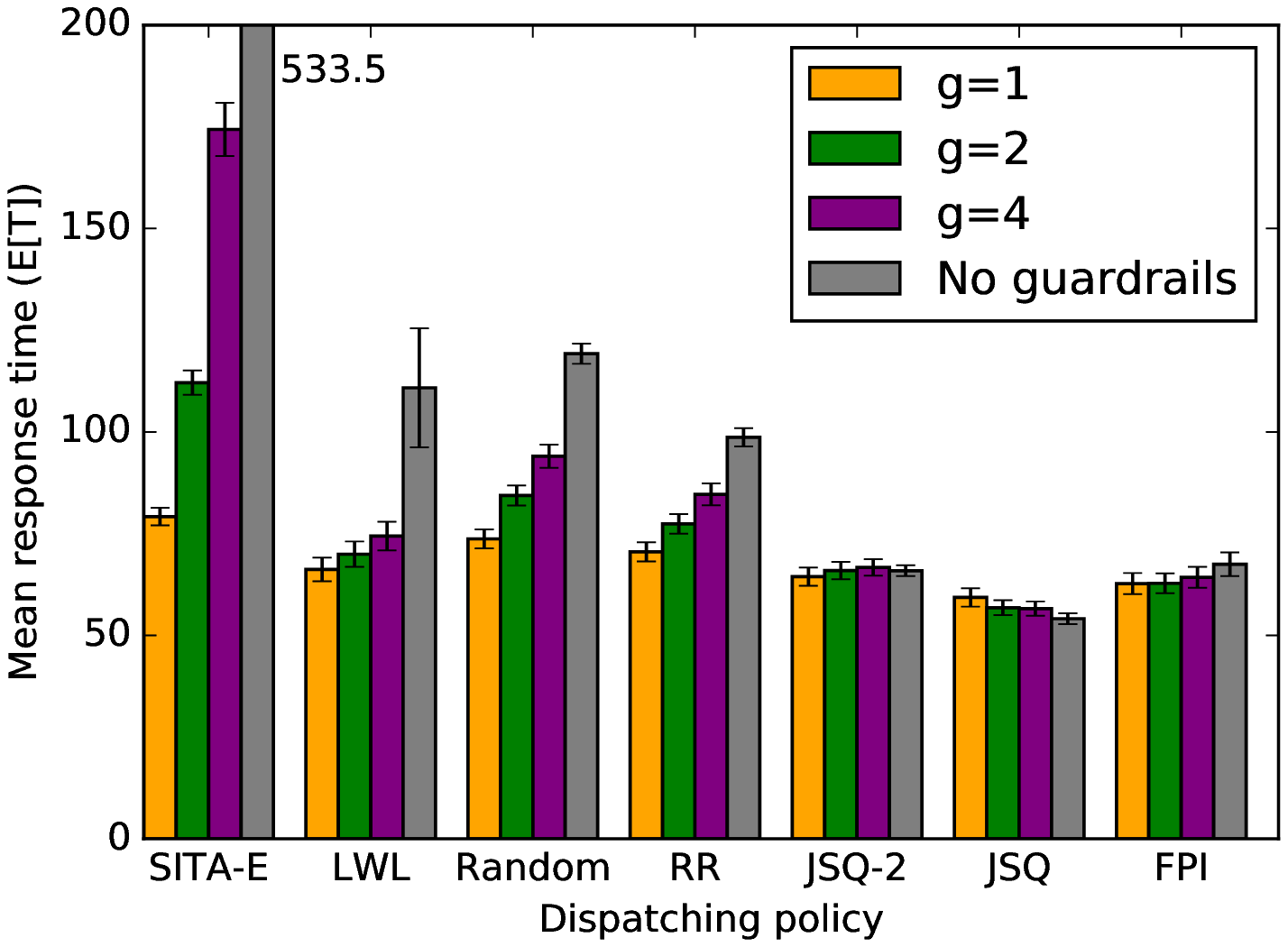}
  \caption{$\rho = 0.98$}
  \vspace{-0.5\baselineskip}
  \end{subfigure}\hfill
  \begin{subfigure}{.48\linewidth}
    \centering
    \includegraphics[width=\linewidth, trim={0 1em 0 0}]{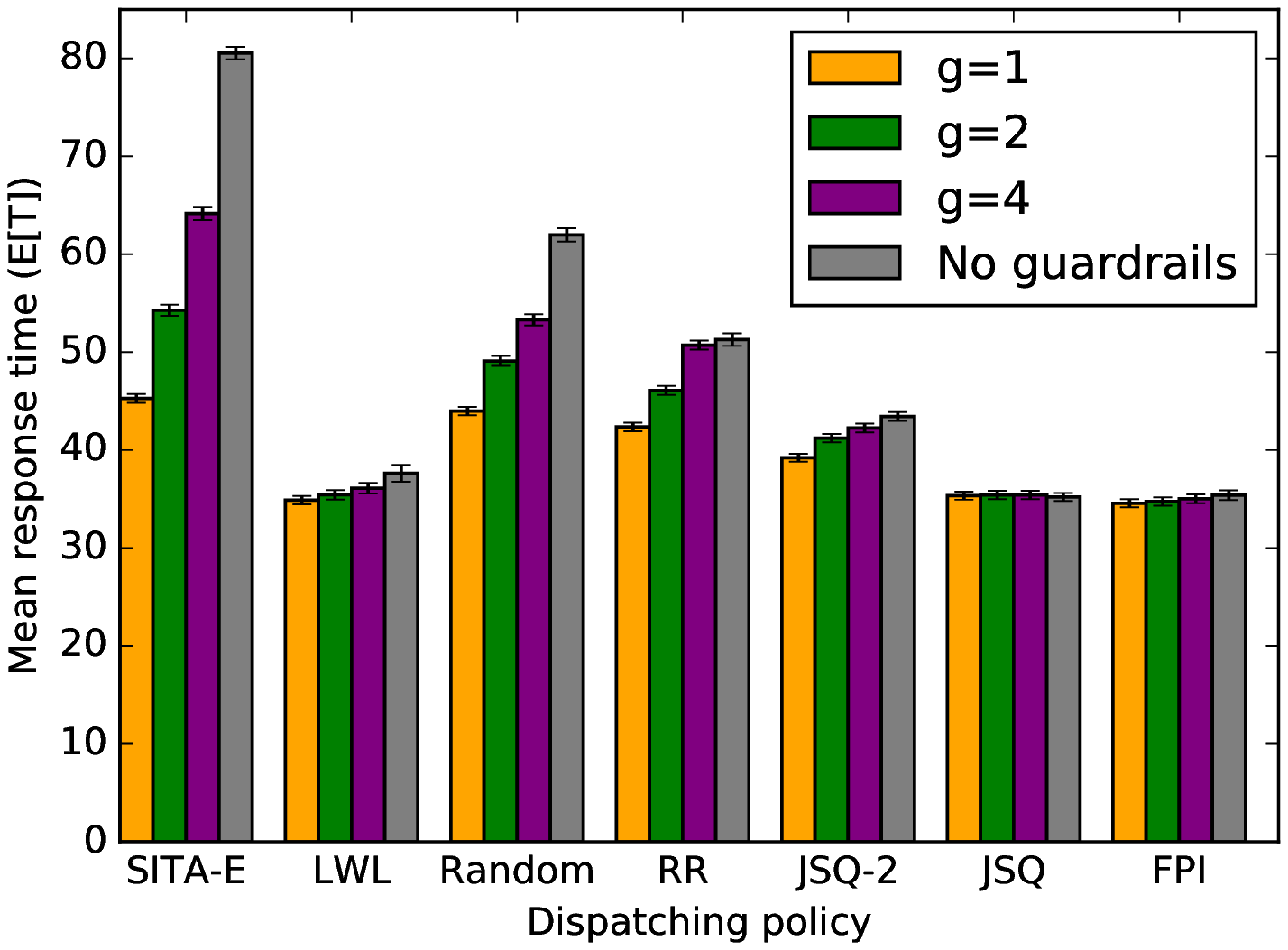}
    \caption{$\rho = 0.80$}
    \vspace{-0.5\baselineskip}
  \end{subfigure}
  \caption{At heavy load $\rho=0.98$, adding guardrails
    significantly reduces the mean response times of
    SITA\=/E, LWL, Random, and RR,
    while leaving the mean response times of
    JSQ, JSQ\=/2, and FPI nearly unchanged.
    At more moderate load $\rho=0.8$,
    adding guardrails significantly reduces
    the mean response times of SITA\=/E, Random, RR, and JSQ\=/2
    while leaving the mean response times of LWL, JSQ, and FPI nearly unchanged.
    The smallest tightness, $g=1$,
    shows the best performance for all policies except JSQ and FPI,
    where it doesn't really matter.
    Simulation uses $k=10$ servers.
    Size distribution shown is Bounded Pareto
    with $\alpha=1.5$ and range $[1, 10^6]$.
    $C^2 \sim 333$. 40 trials simulated, 95\% confidence intervals shown.}
  \label{fig:bp-many-pol-g}
\end{figure*}

\section{Simulation}
\label{sec:simulation}

We have shown that in heavy traffic,
adding guardrails to any dispatching policy gives it
optimal mean response time.
We now turn to investigating loads outside the heavy-traffic regime.
While the mean response time upper bound in Theorem~\ref{thm:guard-bound}
holds for all loads,
it is only tight in the heavy-traffic limit.
We therefore focus on simulation.

We consider the following dispatching policies,
each paired with SRPT scheduling at the servers:
\begin{description}
\item[Random] The policy which dispatches each job to a
  uniformly random server.
\item[Round-Robin (RR)] The policy which dispatches each job
  to the server which least recently received a job.
\item[Least-Work-Left (LWL)] The policy which dispatches each job
  to the server with the least remaining work.
\item[Size-Interval-Task-Assignment (SITA)] The policy which
  classifies jobs into size intervals (small, medium, large, etc.)
  and dispatches all small jobs to one server,
  all medium jobs to another server, etc.
  Specifically, we simulate \textbf{SITA\=/E},
  the SITA policy which chooses the size intervals
  to equalize the expected load at each server.
\item[Join-Shortest-Queue (JSQ)] The policy which dispatches each job
  to the server with the fewest jobs present.
\item[Join-Shortest-Queue-\textit{d} (JSQ\=/\textit{d})] The policy which samples
  $d$ uniformly random servers on each arrival
  and dispatches the job to the server with the fewest jobs present
  among those~$d$.
  We focus on the $d = 2$ case.
\item[First Policy Iteration (FPI)] The first policy iteration heuristic,
  as described by Hyyti\"a et al.~\cite{hyytia-fpi}
  in the setting of dispatching to SRPT servers.
  FPI dispatches each job to
  the server that would be optimal
  if all jobs thereafter were dispatched randomly.
  Hyyti\"a et al.'s derivation of the FPI policy
  assumes that the job size distribution is continuous,
  and specifically that two different jobs almost surely have different sizes.
  As a result, we only implement the FPI policy for the Bounded Pareto distribution,
  shown in Figure~\ref{fig:bp-many-pol-g}.
\end{description}

\subsection{Simulation Results}

\begin{figure*}
  \begin{subfigure}{.48\linewidth}
    \centering
    \includegraphics[width=\linewidth, trim={0 1em 0 0}]{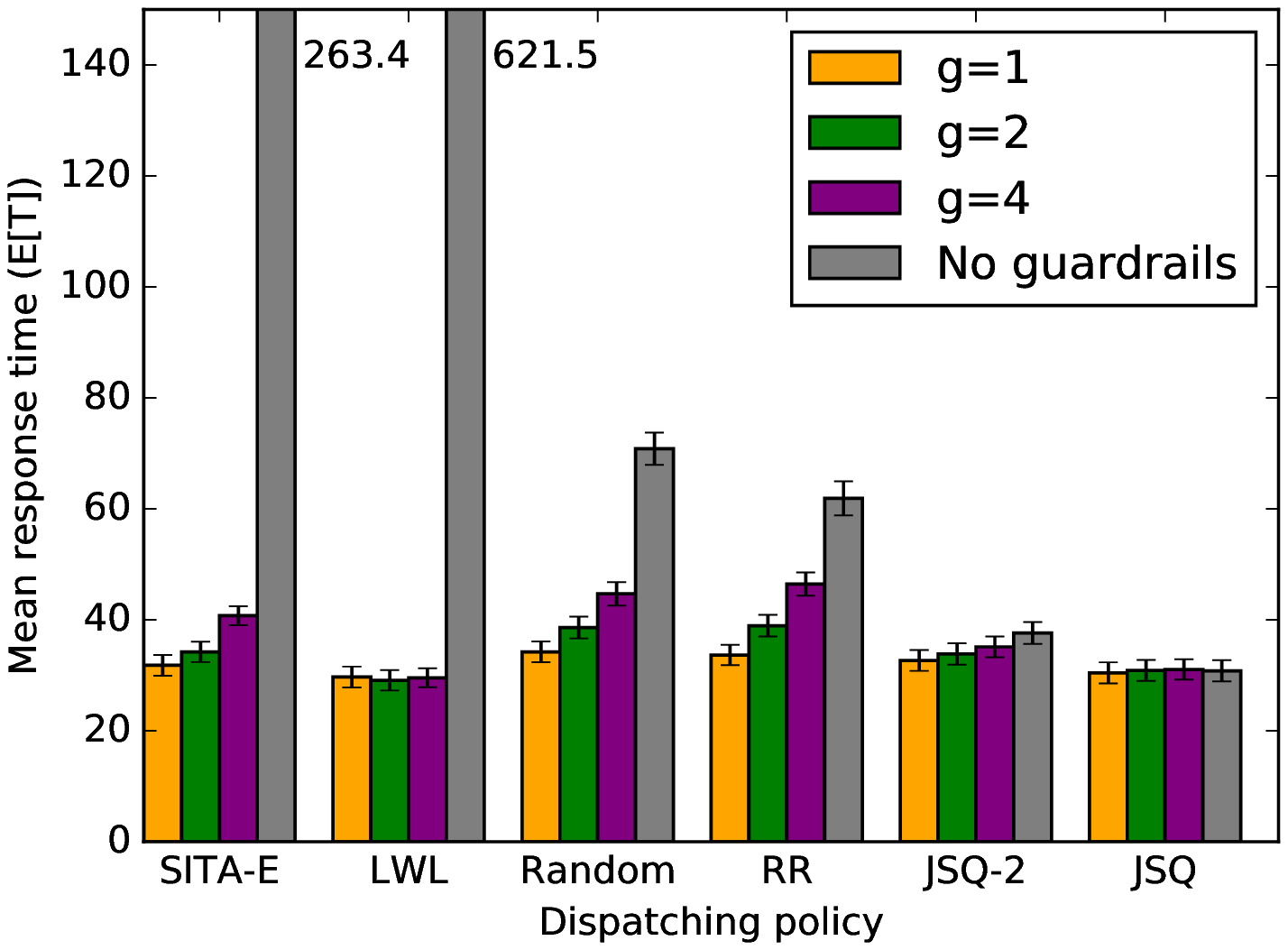}
  \caption{$\rho = 0.98$}
  \vspace{-0.5\baselineskip}
  \end{subfigure}\hfill
  \begin{subfigure}{.48\linewidth}
    \centering
    \includegraphics[width=\linewidth, trim={0 1em 0 0}]{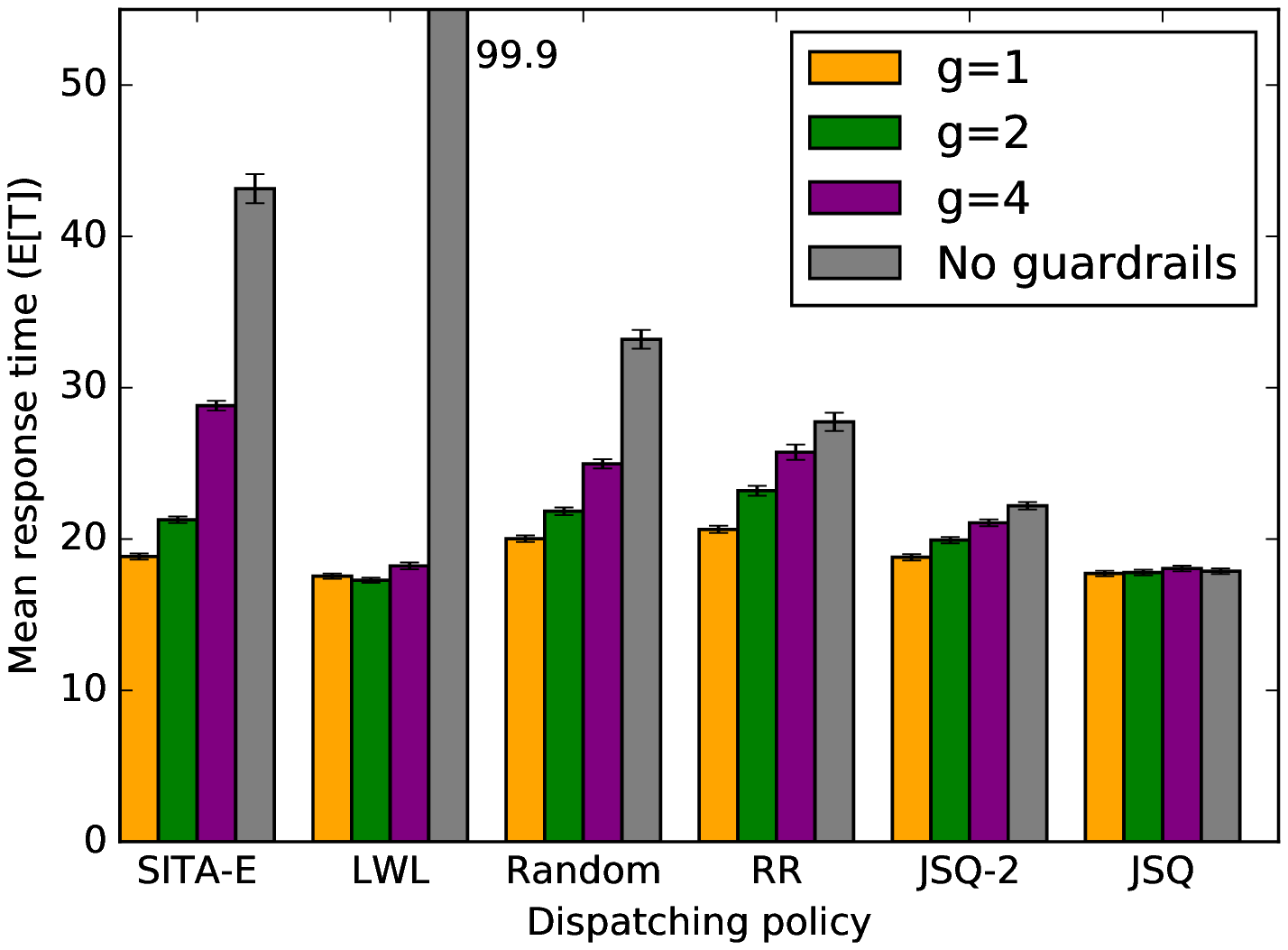}
    \caption{$\rho = 0.80$}
    \vspace{-0.5\baselineskip}
  \end{subfigure}
  \caption{At heavy load $\rho = 0.98$, adding guardrails
    significantly reduces the mean response time of
    SITA\=/E, LWL, Random, and RR,
    while leaving the mean response times of JSQ and JSQ\=/2
    nearly unchanged.
    At more moderate load $\rho=0.8$,
    adding guardrails significantly reduces the mean response times of
    SITA\=/E, LWL, Random, RR, and JSQ\=/2,
    while leaving the mean response time of JSQ nearly unchanged.
    The smallest tightness, $g=1$,
    shows the best performance for all policies except JSQ,
    where it doesn't really matter.
      Simulation uses $k=10$ servers.
      Size distribution shown is Bimodal with
      99.95\% size 1 jobs and 0.05\% size 1000 jobs.
      $C^2 \sim 221$. 40 trials simulated, 95\% confidence intervals shown.}
  \label{fig:bm-many-pol-g}
\end{figure*}

Figures~\ref{fig:bp-many-pol-g} and~\ref{fig:bm-many-pol-g} show the
mean response time under all of the above dispatching policies
with SRPT scheduling at the servers.
We omit the FPI policy from Figure~\ref{fig:bm-many-pol-g},
because that simulation's job size distribution is not continuous,
and Hyyti\"a et al. do not derive the FPI policy for such distributions~\cite{hyytia-fpi}.
We also show $95\%$ confidence intervals for each mean response time.
We consider two different job size distributions:
a Bounded Pareto distribution (see Figure~\ref{fig:bp-many-pol-g})
and a Bimodal distribution (see Figure~\ref{fig:bm-many-pol-g}).
In each case we show (a)~very heavy traffic ($\rho = 0.98$)
and (b)~more moderate traffic ($\rho = 0.8$).
We augment each dispatching policy with guardrails of varying tightness,
$g=1, 2,$ and $4$.

The high-level message seen in Figures~\ref{fig:bp-many-pol-g}
and~\ref{fig:bm-many-pol-g} is that adding guardrails
to dispatching policies can greatly reduce their response times,
even at more moderate loads.
Simple dispatching policies
like Random and RR improve by $15-40\%$ when $\rho = 0.8$
and $30-50\%$ when $\rho = 0.98$, in the figures shown.
Other policies, like LWL and SITA\=/E,
show even more dramatic improvement for certain job size distributions.
We find using tightness $g=1$ is generally best.

The FPI heuristic of Hyyti\"a et al.~\cite{hyytia-fpi} performs about
equally well with or without guardrails.
Figure~\ref{fig:bp-many-pol-g} shows that adding guardrails to FPI
yields a slight reduction in mean response time at $\rho = 0.98$,
and has essentially no effect at $\rho = 0.8$.
The FPI heuristic performs well in simulation,
but its only theoretical guarantee is it outperforms Random.
Applying guardrails
guarantees optimal mean response time in heavy traffic,
while maintaining or improving performance in simulation.

We observe that the JSQ dispatching policy
performs well even without guardrails.
In fact, guardrails can be seen as helping
all the other dispatching policies to improve their performance
to approach JSQ.
We do not know of any guarantees on JSQ's performance with \Srpt/ servers,
even under heavy traffic,
unless we add guardrails to JSQ.
Figures~\ref{fig:bp-many-pol-g} and~\ref{fig:bm-many-pol-g}
show that adding guardrails to JSQ does not affect its performance much.

In Appendix~\ref{app:simulation}, we simulate guarded policies
under a variety of alternative system conditions.
We simulate systems with more servers, systems with lighter load
and systems with different job size distributions.

\subsection{Simulation Discussion and Intuition}
Recall from Section~\ref{sec:guard-def} the intuition behind guardrails:
Guardrails force the dispatching policy to spread out
small jobs
across all of the servers.
Guardrails thereby
ensure that the maximum possible number of servers are working
on the smallest jobs available.

Let us consider this intuition in light of each of the dispatching policies.
SITA\=/E does the opposite of spreading small jobs:
It clumps the smallest jobs onto the same server.
Therefore, G-SITA\=/E shows a massive improvement.
In particular, we want $g$ to be
as low as possible ($g=1$),
corresponding to the greatest guardrail control,
to prevent SITA\=/E from doing what it was designed to do.

Random and RR are better at spreading jobs naturally,
but they still make mistakes.
In particular, Random and RR do not differentiate
between jobs of different sizes
and they do not observe the state of the servers,
so they only spread out the small jobs by chance.
As a result, G-Random and G-RR show sizable improvements.
The tightest guardrails ($g=1$)
increase the spread of the small jobs the most,
and hence show the most improvement.

LWL does not spread out the small jobs.
In fact, one huge job at a server can keep away all of the small jobs
for a long time.
However, LWL is so efficient at using its servers
that one does not experience its shortcomings
unless both load and job size variability are high.
Under those circumstances, LWL has very high mean response times.
At load $\rho = 0.98$
with the Bimodal size distribution shown in Figure~\ref{fig:bm-many-pol-g},
LWL's mean response time is 7 times worse
than that of Random.
Guardrails are particularly effective in these situations
where LWL fails because they force small jobs onto the servers
with one large job.
The tightest guardrails ($g=1$)
force small jobs onto those servers most aggressively
and hence show the lowest mean response time
in Figures~\ref{fig:bp-many-pol-g} and~\ref{fig:bm-many-pol-g}.

\subsection{Comparing Simulation to Analytical Bounds}

In Theorem~\ref{thm:guard-bound}
we established an analytical upper bound on the mean response time
of any guarded policy.
This bound is tight in the limit as load $\rho \rightarrow 1$,
and implies the heavy traffic optimality of any guarded policy.

However, this bound is not tight under the more moderate loads simulated in this section.
Under the system conditions shown in Figure~\ref{fig:bp-many-pol-g},
at load $\rho = 0.8$ and guardrail tightness $g=1$,
Theorem~\ref{thm:guard-bound} implies that any guarded policy has
mean response time at most $350$,
and that at load $\rho = 0.98$ the mean response time is at most $700$.
The actual performance of guarded policies is much better than this,
as shown in Figure~\ref{fig:bp-many-pol-g}.
Tightening our bound is a potential direction for future research.

\section{Practical Considerations}
\label{sec:practical}

We now discuss several useful properties of guardrails
that help when implementing them in practical systems.
We also extend guardrails to cover a broader range of applications.

\subsection{Robustness to Network Delays}
\label{sec:network-delays}

Guardrails are relatively simple to implement:
the dispatcher stores work counters for each rank and each server,
increasing the appropriate work counter whenever it dispatches a job
(see Algorithm~\ref{alg:guarded}).
For the most part,
the dispatcher does not need to monitor the precise state of each server.
The only exception is that
whenever a server becomes empty, the server \emph{resets},
which decreases all of the dispatcher's work counters for that server.
As we will explain shortly, this complicates the implementation,
particularly in settings with network delays.

Fortunately, resets are optional for the purposes of heavy-traffic optimality
(see Remark~\ref{rmk:resets-optional}).
However, resets are still desirable because
they help decrease response time at lower loads.
Specifically, resets ensure that a guarded policy is always allowed to
dispatch jobs to empty servers.
We thus do not want to ignore resets entirely.

To implement resets without
the dispatcher needing to track the remaining work at each server at all times,
servers can send ``reset messages'' to the dispatcher when they become empty.
This works well so long as messages do not experience network delays,
because our analysis
(see Lemmas~\ref{lem:work-balance} and~\ref{lem:first-bound})
assumes that servers only reset when they are empty,
which might not be the case if a reset message is delayed.

In practice, reset messages may well experience network delays,
To handle delays, the dispatcher should ignore reset messages
from servers that might not be empty.
One protocol for doing so is:
\begin{itemize}
\item
  The dispatcher stores, for each server~$s$,
  a hash of all the job IDs sent to~$s$.
\item
  Each server~$s$ stores a hash of all the job IDs it has received.
\item
  When a server becomes empty,
  it sends a reset message to the dispatcher
  which contains the currently stored hash.
\item
  When the dispatcher receives a reset message from server~$s$,
  it resets~$s$ if
  the reset message's hash matches the stored hash for~$s$.
  Otherwise, the dispatcher ignores the reset message.
\end{itemize}

\subsection{Multiple Dispatchers}
\label{sec:multiple-dispatchers}

Many large load balancing systems in practice have multiple dispatchers,
each of which sends jobs to the same set of servers.
Communication between the dispatchers may be limited,
in which case they each have to make dispatching decisions independently.
Fortunately, in systems with multiple dispatchers,
it suffices to have each dispatcher independently implement guardrails.
As explained below,
we obtain the same theoretical guarantees for each of the following:
\begin{itemize}
\item
  A system with $d$ dispatchers,
  each independently satisfying guardrails with tightness~$g$.
\item
  A system with a single dispatcher
  satisfying guardrails with tightness~$dg$.
\end{itemize}
Guardrails thus guarantee heavy-traffic optimality
for systems with any constant number of dispatchers.

To see why it suffices to implement guardrails separately for each dispatcher,
consider a system with $d$ dispatchers.
Suppose each dispatcher separately keeps ``local'' guardrail work counters,
which only track arrivals to that dispatcher,
and implements guardrails with tightness~$g$.
We can also imagine what the ``global'' guardrail work counters,
which track all arrivals at all dispatchers,
would look like,
even though there is no physical device storing them.
We ask: given that the local counters have tightness~$g$,
what is the tightness of the global counters?
Consider
the local and global rank~$r$ guardrail work counters
for two servers $s$ and~$s'$.
Each dispatcher's local counter pair has difference at most $gc^{r+1}$
(see Definition~\ref{def:guardrails}),
and there are $d$ dispatchers,
so the global counter pair has difference at most $dgc^{r+1}$.
This means the global counters stay within tightness~$dg$.

Systems with multiple dispatchers tend to be large systems
in which network delays are non-negligible.
The reset protocol from Section~\ref{sec:network-delays}
can be easily adapted to multiple dispatchers by having each server
store a separate hash of job IDs for each dispatcher.

\subsection{Scheduling Policies other than \Srpt/}

We have shown that guarded dispatching policies provide theoretical guarantees
and good empirical performance
for load balancing systems using \Srpt/ scheduling at the servers.
However, in some settings it is impossible to use \Srpt/.
For example, network hardware often allows scheduling using
only finitely many priority classes,
in which case \Srpt/ can only be
approximated~\cite{ousterhout-homa, mor-web-server}.
Systems may also choose a non-\Srpt/ scheduling policy for other reasons,
such as fairness concerns~\cite{wierman-fairness}.

Guardrails are sometimes suitable
even when the servers are using a scheduling policy other than \Srpt/.
In particular, we can extend our theoretical guarantees
to many preemptive size-based scheduling policies that favor small jobs.
We have already proven such a guarantee for the \Prio/ policy
(see Theorem~\ref{thm:guard-srpt}).
Using bounds proved by Wierman et al.~\cite{wierman-smart},
one can extend our results to
all policies in their SMART class,
which includes Preemptive-Shortest-Job-First (\Psjf/) and
Shortest-Processing-Time-Product (SPTP, also known as RS).

Guardrails can also provide guarantees for
size-based policies with finitely many priority classes,
which are used in some computer systems
to approximate \Srpt/~\cite{ousterhout-homa, mor-web-server}.
In this setting, each priority class corresponds to an interval of job sizes.
Here it is most natural to use a slightly modified version of guardrails:
a guarded policy is one ensuring that for any class~$i$,
the maximum difference between two servers' class~$i$ work counters
never exceeds the upper bound of class~$i$'s size interval.
If the job size distribution is bounded,
these modified guardrails guarantee a mean response time bound
analogous to Theorem~\ref{thm:guard-bound}.
This implies that in the heavy traffic limit,
the system's performance approaches that of one large server
using the same scheduling policy.

So far, we have only considered policies
that use job size information to favor small jobs.
This is the setting in which guardrails are most likely to be effective.
We conjecture that guardrails might also be useful
for servers using PS or Foreground-Background (FB) scheduling,
as these policies also tend to favor small jobs,
so they may benefit from spreading out small jobs across the servers.

\subsection{Heterogeneous Server Speeds}

We have thus far assumed that all servers in the system
have the same speed,
but this is not always the case.
Fortunately, guardrails can be adapted to systems
with heterogeneous server speeds.
The key is to track each server's guardrail work counter $G_s(t)$ in units of time.
That is, when we dispatch job of size~$x$ to a server with speed $\mu$,
we increase the server's guardrail work counter by $x/\mu$.
It is simple to generalize our response time bound in Theorem~\ref{thm:guard-bound}
to this setting by multiplying the bound's last term by $\mu_{\max}/\mu_{\min}$,
where $\mu_{\min}$ and $\mu_{\max}$ are
the minimal and maximal server speeds, respectively.
This implies that any guarded policy paired with \Srpt/ service
is heavy-traffic optimal with heterogeneous servers.

\section{Conclusion}

We introduce \emph{load balancing guardrails},
a technique for augmenting dispatching policies
that ensures low response times
in load balancing systems using SRPT scheduling at the servers.
We prove that guardrails guarantee optimal mean response time in heavy traffic,
and we show empirically that guardrails reduce mean response time
across a range of loads.
Moreover, guardrails are simple to implement
and are a practical choice for large load balancing systems,
including those with multiple dispatchers and network delays.

One direction for future work could address a limitation of guardrails:
they require the dispatcher to know each job's exact size.
Many computer systems only have access to noisy job size estimates
or have no size information at all.
When exact size information is not available,
minimizing mean response time becomes much more complex,
as it is not even clear what scheduling policy should be used at the servers.
It is possible that a variation of guardrails could be
used to create good dispatching policies
when using the celebrated Gittins index scheduling
policy~\cite{gittins-book, aalto-m/g/1-gittins} at the servers.

Our analysis of guardrails constitutes
the first closed-form mean response time bound
for load balancing systems with general job size distribution
and complex dispatching and scheduling policies.
However, the bound is only tight in the heavy-traffic limit.
Developing better analysis tools for the light traffic case
remains an important open problem.

\begin{acks}
  We thank Gustavo de Veciana and the anonymous referees
  for their helpful comments.
  This research was supported by
  NSF-\grantnum{nsf}{XPS-1629444},
  NSF-\grantnum{nsf}{CSR-180341},
  and a 2018 Faculty Award from
  \grantsponsor{microsoft}{Microsoft}{%
    https://www.microsoft.com/en-us/research/academic-programs/}.
  Additionally, Ziv Scully was supported by an
  \grantsponsor{arcs}{ARCS Foundation}{https://www.arcsfoundation.org}
  scholarship and the \grantsponsor{nsf}{NSF}{https://www.nsf.gov} GRFP under
  Grant Nos.~\grantnum{nsf}{DGE-1745016} and~\grantnum{nsf}{DGE-125222}.
\end{acks}

\bibliographystyle{ACM-Reference-Format}
\bibliography{imd}

\appendix

\section{Matuszewska Index}
\label{app:matuszewska}
The optimality results in this paper,
such as Theorem~\ref{thm:guard-srpt},
assume that the job size distribution~$X$ is not too heavy-tailed.
Specifically, we assume that either $X$ is bounded,
or that the \emph{upper Matuszewska index} of the tail of $X$
is less than~$-2$.
This is slightly stronger than assuming that $X$ has finite variance.
The formal definition of the upper Matuszewska is the following.

\begin{definition}
  Let $f$ be a positive real function. The \emph{upper Matuszewska index} of $f$,
  written $M(f)$,
  is the infimum over $\alpha$
  such that there exists a constant $C$
  such that for all $\gamma > 1$,
  \[\lim_{x \to \infty} \frac{f(\gamma x)}{f(x)} \le C \gamma^\alpha.\]
  Moreover, for all $\Gamma>1$, the convergence $x\to\infty$ must be uniform
  in $\gamma \in [1, \Gamma]$.
\end{definition}
The condition $M(\bar{F}_X)<-2$, where $\bar{F}_X$ is the tail of $X$,
is intuitively close to saying that $F_X(x) \le Cx^{-2-\epsilon}$
for some constant $C$ and some $\epsilon > 0$.
Roughly speaking, this means that $X$ has a lighter tail
than a Pareto distribution with $\alpha = 2$.

\section{Lemma~\ref{lem:prio-psjf}}
\label{app:prio-psjf}
\begin{lemma}
  For any job size distribution,
  the mean response time of a single-server \Prio/ system
  is no more than $c$ times the mean response time of a
  single-server \Psjf/ system:
  \label{lem:prio-psjf}
  \begin{equation*}
    \E{T}^{\Prio/} \le (c + 2\sqrt{c-1})\E{T}^{\Psjf/}.
  \end{equation*}
\end{lemma}
\begin{proof}
  We will consider a new random variable, $D$, the delay due to a job.
  $D$ is defined for scheduling policies that assign every job a fixed priority,
  like \Prio/ and \Psjf/.
  For a given job $j$ of size $x$, $D_j$ is
  \begin{itemize}
    \item
      the amount $j$ delays other jobs,
      namely $x$ times the number of jobs with lower priority than $j$
      in the system when $j$ arrives,
    \item
      plus the amount other jobs that arrived before $j$ delay $j$,
      namely the total remaining size of jobs with higher priority than $j$
      in the system when $j$ arrives,
    \item
      plus $j$'s size.
  \end{itemize}

  Note that the response time of a job $\ell$ is equal to $\ell$'s size,
  plus the amount $\ell$ is delayed by jobs that arrived before $\ell$,
  plus the amount $l$ is delayed by jobs that arrive after $\ell$.
  Each of those amounts of time is accounted for in the delay of exactly one job.
  As a result, the sum of the delays of the jobs in a busy period
  equals the sum of the response times of those jobs.
  Therefore, in steady state, mean response time and mean delay are equal:
  \begin{equation*}
    \E{T} = \E{D}.
  \end{equation*}

  Therefore, it suffices to show that
  \begin{equation*}
    \E{D}^{\Prio/} \le (c+2\sqrt{c-1})\E{D}^{\Psjf/}.
  \end{equation*}

  Let us consider a pair of coupled systems receiving the same arrivals:
  A single-server system with \Psjf/ scheduling,
  and a single-server system with \Prio/ scheduling.

  Note that \Psjf/ and \Prio/ both prioritize all jobs of lower ranks
  over all jobs of higher ranks.
  As a result, both coupled systems will server jobs of the same ranks
  at the same times,
  and will always have the same amount of remaining work
  of each rank.

  Let us consider the expected delay due to a particular ``tagged'' job $j$,
  arriving to a steady-state system.
  The expected delay due to $j$ is equal to each system's mean delay
  by the PASTA property~\cite{wolff-pasta}.
  Let $j$ be a job of size $x$ with rank $r = \log_c x$.

  The delay due to $j$, $D_j$, is a summation over each job in the system
  at the moment $j$ arrives.
  Let $D_j^q$ be the delay caused by the interaction of $j$ and jobs
  of rank $q$ that are in the system when $j$ arrives,
  including $j$'s size in $D_j^r$.
  Then we can write $D_j$ in terms of the $D_j^q$s:
  \begin{equation*}
    D_j = \sum_{q = -\infty}^{\infty} D_j^q.
  \end{equation*}
  Therefore, it suffices to show for all ranks $q$ that
  \begin{equation}
    \label{eqn:delay-rank}
    \E{D_j^q}^{\Prio/} \le (c + 2\sqrt{c-1}) \E{D_j^q}^{\Psjf/}.
  \end{equation}

  Let $X_{j'}$ denote the original size of a job $j'$,
  and let $R_{j'}$ denote the remaining size of $j'$.
  Let $J^q$ denote the set of jobs of rank $q$ in the system
  at the time $j$ arrives.
  Let $N^q$ denote the number of jobs of rank $q$
  in the system at the time $j$ arrives.

  We now consider three cases: $q < r$, $q = r$, and $q > r$.

  \begin{case}[$q < r$]
    Because $q < r$,
    all jobs in rank $q$ have higher priority than $j$.
    As a result, under both \Psjf/ and \Prio/,
    $D_j^q$ is equal to the total remaining size of jobs of rank $q$:
    \[D_j^q = \sum_{j' \in J^q} R_{j'}.\]
    As noted above, this is equal in the two systems due to the coupling.
    This proves (\ref{eqn:delay-rank}) in this case.
  \end{case}

  \begin{case}[$q = r$]
    Because \Prio/ uses First-Come-First-Served scheduling within a rank,
    $D_j^r$ in the \Prio/ system is equal to
    the total remaining size of jobs of rank $q$:
    \[{D_j^r}^{(\Prio/)} = \sum_{j' \in {J^r}^{(\Prio/)}} \mathllap{R_{j'}.}\]
    In contrast, $D_j^r$ in the \Psjf/ system is equal to the
    total remaining size of jobs of rank $r$ with size at most $x$,
    plus $x$ times the number of jobs of rank $r$ with size
    more than $x$:
    \begin{equation*}
      {D_j^r}^{(\Psjf/)} = \sum_{j' \in {J^r}^{(\Psjf/)}|X_{j'} \le x}\mathllap{R_{j'}\quad} + \sum_{j' \in {J^r}^{(\Psjf/)}|X_{j'} > x} \mathllap{x.\ \quad}
    \end{equation*}
    Noting that $x \ge c^r$ and the remaining size of any job of rank $r$ is at most
    $c^{r+1}$, we can lower bound ${D_j^r}^{(\Psjf/)}$:
    \begin{align*}
      \ifwidecol{}{\MoveEqLeft} {D_j^r}^{(\Psjf/)} \ifwidecol{}{\\}
      &\ge \sum_{j' \in {J^r}^{(\Psjf/)}|X_{j'} \le x} \mathllap{R_{j'}\quad}
        + \sum_{j' \in {J^r}^{(\Psjf/)}|X_{j'} > x} \mathllap{c^r\,\quad}\\
      &\ge \sum_{j' \in {J^r}^{(\Psjf/)}|X_{j'} \le x} \mathllap{R_{j'}\quad}
        + \sum_{j' \in {J^r}^{(\Psjf/)}|X_{j'} > x} \mathllap{\frac{R_{j'}}{c}\!\quad}\\
      &\ge \frac{1}{c}\sum_{j' \in {J^r}^{(\Psjf/)}} \mathllap{R_{j'}.\!\!\!}
    \end{align*}
    As noted above, $\sum_{j' \in J^q} R_{j'}$ is equal in both systems.
    As a result,
    \[\E{D_j^q}^{\Prio/} \le c \E{D_j^q}^{\Psjf/},\]
    which proves (\ref{eqn:delay-rank}) in this case.
  \end{case}

  \begin{case}[$q > r$]
    Because $q > r$, all jobs in rank $q$
    have lower priority that $j$.
    As a result, in both systems,
    \[D_j^q = x N^q.\]
    Also, note that $x$ is independent of $N^q$.
    Therefore, we simply need to show that
    \[\E{N^q}^{\Prio/} \le (c+2\sqrt{c-1}) \E{N^q}^{\Psjf/}.\]
    However, it is possible for there to be twice as many jobs of rank $q$
    in the \Prio/ system as in the \Psjf/ system, regardless of the value of $c$.
    In particular, there could
    be two rank $q$ jobs in the \Prio/ system and one rank $q$ job in the \Psjf/ system,
    if one of the rank $q$ jobs in the \Prio/ system has very little remaining size.

    Let $j^\textrm{old}$ be the oldest job of rank~$q$ in the \Prio/ system
    at a given time.
    Because \Prio/ serves jobs in FCFS order,
    only $j^\textrm{old}$ has been processed,
    so only $j^\textrm{old}$ can have remaining size under $c^q$.
    Therefore, we can bound the total remaining size of the rank $q$ jobs
    in the \Prio/ system:
    \[\sum_{j' \in {J^q}^{(\Prio/)}} \mathllap{R_{j'}\ \ } \ge R_{j^\textrm{old}} + c^q  ({N^q}^{(\Prio/)}-1).\]
    Likewise, we can bound the total remaining size of the rank $q$ jobs
    in the \Psjf/ system:
    \[\sum_{j' \in {J^q}^{(\Psjf/)}} \mkern-20mu R_{j'} \le c^{q+1}  {N^q}^{(\Psjf/)}.\]
    Using the fact that $\sum_{j' \in J^q} R_{j'}$ is equal in both systems gives us
    \[R_{j^\textrm{old}} + c^q  ({N^q}^{(\Prio/)}-1) \le c^{q+1}  {N^q}^{(\Psjf/)},\]
    which rearranges to
    \[{N^q}^{(\Prio/)} \le c  {N^q}^{(\Psjf/)} + 1 - \frac{R_{j^\textrm{old}}}{c^q}.\]
    This implies ${N^q}^{(\Prio/)} \le c  {N^q}^{(\Psjf/)} + 1$,
    which allows us to relate the expected numbers of jobs in the systems:
    \ifwidecol{\begin{equation}}{\begin{align}}
      \ifwidecolignorealigny
      \notag
      \MoveEqLeft \E{{N^q}^{(\Prio/)}}\\
      \notag
      \ifwidecol{}{&}\le c \E{{N^q}^{(\Psjf/)}}\\
      \label{eqn:hard-case}
      \ifwidecol{}{&\qquad} + \Pr \{{N^q}^{(\Prio/)} > c {N^q}^{(\Psjf/)}\}.
    \ifwidecol{\end{equation}}{\end{align}}

    Recall that ${N^q}^{(\Prio/)}$ and ${N^q}^{(\Psjf/)}$ are both integers.
    This means that if ${N^q}^{(\Prio/)} > c {N^q}^{(\Psjf/)}$, then
    \[{N^q}^{(\Prio/)} \ge {N^q}^{(\Psjf/)} + 1.\]
    As a result, if ${N^q}^{(\Prio/)} > c {N^q}^{(\Psjf/)}$, then
    \[{N^q}^{(\Psjf/)} + 1 \le c  {N^q}^{(\Psjf/)} + 1 - \frac{R_{j^\textrm{old}}}{c^q},\]
    meaning that
    \begin{equation}
      \label{eqn:rjold-small}
      R_{j^\textrm{old}} \le (c-1)c^q{N^q}^{(\Psjf/)}.
    \end{equation}
    Therefore, we will show that either
    $\Pr \{{N^q}^{(\Prio/)} > c  {N^q}^{(\Psjf/)}\}$
    is small or
    $\E{{N^q}^{(\Psjf/)}}$
    is large.
    In either case, \eqref{eqn:hard-case} will imply
    the desired bound \eqref{eqn:delay-rank}.

    Let us condition on $R_{j^\textrm{old}}$. Because $j$ is a Poisson arrival,
    $j$ sees a time-average state of $R_{j^\textrm{old}}$. The (stochastically)
    smallest this
    distribution can be is the uniform distribution on $[0, c^q]$, because
    the original size of a rank $q$ job is at least $c^q$.
    In particular, for any $\ell$,
    \begin{equation*}
      \Pr \{ R_{j^\textrm{old}} < \ell\} \le \frac{\ell}{c^q}.
    \end{equation*}

    Let $\rho_q$ be the probability that $j$ sees a rank $q$ job in the system
    on arrival.
    Let $m$ be the largest integer such that
    \[\Pr \{{N^q}^{(\Prio/)} > c  {N^q}^{(\Psjf/)}\} \ge m (c-1)\rho_q .\]
    Then (\ref{eqn:rjold-small}) implies
    \begin{equation}
      \label{eq:pr-lower-bound}
      \Pr \{R_{j^\textrm{old}} \le c^q(c-1){N^q}^{(\Psjf/)}\} \ge m (c-1) \rho_q.
    \end{equation}

    Regardless of the correlation between ${N^q}^{(\Psjf/)}$ and $R_{j^\textrm{old}}$,
    we must have
    \begin{equation}
      \label{eq:pr-upper-bound-1}
      \Pr \{ {N^q}^{(\Psjf/)} \ge m\} \ge (c-1) \rho_q.
    \end{equation}
    This is because if ${N^q}^{(\Psjf/)} \le m-1$ on a particular arrival,
    and also ${N^q}^{(\Prio/)} > c  {N^q}^{(\Psjf/)}$,
    then by~(\ref{eqn:rjold-small}),
    \begin{equation}
      \label{eq:rjold}
      R_{j^\textrm{old}} \le (c-1)c^q(m-1).
    \end{equation}
    But an arrival only sees job~$j^{\textrm{old}}$ in the system at all
    with probability $\rho_q$.
    Moreover, conditional on seeing a job~$j^{\textrm{old}}$ at all,
    the arrival observes \eqref{eq:rjold} with probability at most $(c-1)(m-1)$,
    because every rank~$q$ job has size at least~$c^q$.
    This means
    \[\Pr \{R_{j^\textrm{old}} \le (c-1)c^q(m-1)\} \le (m-1)(c-1) \rho_q,\]
    which together with \eqref{eq:pr-lower-bound} implies \eqref{eq:pr-upper-bound-1}.

    By a similar argument as \eqref{eq:pr-upper-bound-1},
    \[\Pr \{ {N^q}^{(\Psjf/)} \ge m-z\} \ge (z + 1)(c-1) \rho_q\]
    for any integer $z < m$.

    In addition, if there is a job in the \Prio/ system then there is a job
    in the \Psjf/ system:
    \[\Pr \{ {N^q}^{(\Psjf/)} \ge 1 \} \ge \rho_q.\]
    We can combine these bounds on the probability of ${N^q}^{(\Psjf/)}$ taking
    specific values to a derive
    a bound on its expectation:
    \begin{align*}
      \E{{N^q}^{(\Psjf/)}}
      &\ge \biggl(\sum_{z=0}^{m-1} (m-z)(c-1)\rho_q\biggr) + (\rho_q-\rho_q m(c-1))\\
      &= \biggl(\biggl(\sum_{z=0}^{m-1} (m-z-1)(c-1)\biggr) + 1\biggr)\rho_q\\
      &= \biggl(1 + \frac{m}{2}(m-1)(c-1)\biggr) \rho_q.
    \end{align*}

    We are now ready to show that either
    $\Pr \{{N^q}^{(\Prio/)} > c  {N^q}^{(\Psjf/)}\}$
    is small or
    $\E{{N^q}^{(\Psjf/)}}$
    is large
    by bounding their ratio for any value of $m$:
    \begin{align*}
      \ifwidecol{}{\MoveEqLeft}
      \frac{\Pr \{{N^q}^{(\Prio/)} > c  {N^q}^{(\Psjf/)}\}}{\E{{N^q}^{(\Psjf/)}}}
      \ifwidecol{}{\\}
      &\le \frac{(m+1)(c-1)\rho_q}{(1 + \frac{m}{2}(m-1)(c-1))\rho_q}\\
      &=\frac{m+1}{\frac{1}{c-1} + \frac{m}{2}(m-1)}.
    \end{align*}
    This expression is maximized when
    \[m = \sqrt{\frac{2c}{c-1}} - 1,\]
    so
    \begin{align*}
      \ifwidecol{}{\MoveEqLeft}
      \frac{\Pr \{{N^q}^{(\Prio/)} > c  {N^q}^{(\Psjf/)}\}}{\E{N^q}^{(\Psjf/)}}
      \ifwidecol{}{\\}
      &\le \frac{\sqrt{\frac{2c}{c-1}}}{\frac{1}{c-1} + (\sqrt{\frac{c}{c-1}} - \frac{1}{2})(\sqrt{\frac{2c}{c-1}} - 2)}\\
      &=\frac{1}{\sqrt{\frac{2c}{c-1}} - \frac{3}{2}}\\
       &\le 2\sqrt{c-1}.
    \end{align*}
    where the final bound holds due to the fact that $1 < c \le 2$,
    by the definition of $c$.
    Combining this result with (\ref{eqn:hard-case}) yields
    \[\E{{N^q}^{(\Prio/)}} \le  (c + 2\sqrt{c-1}) \E{{N^q}^{(\Psjf/)}},\]
    which is (\ref{eqn:delay-rank}), as desired.
    \qedhere
   \end{case}
  \end{proof}

\section{Lemma~\ref{lem:lin-bound}}
\label{app:lin-bound}

\begin{lemma}
  \label{lem:lin-bound}
  For any size distribution~$X$ which is either
  (i)~bounded or
  (ii)~unbounded with tail having upper Matuszewska index\matuszewska{} less than~$-2$,
  \begin{equation*}
    \lim_{\rho \to 1}\frac{\ln^2\frac{1}{1-\rho}}{\E{T}^{\Srpt/}} = 0.
  \end{equation*}
\end{lemma}

\begin{proof}
  Lin et al.~\cite{Lin} show in their Theorem~1 that if $X$ is bounded,
  then
  \[\E{T}^{\Srpt/} = \theta \biggl(\frac{1}{1-\rho} \biggr),\]
  proving case~(i).
  They also show in their Theorem~2 that if the Matuszewska index
  of the tail of $X$ is less than~$-2$, then
  \[\E{T}^{\Srpt/} = \theta \biggl(\frac{1}{(1-\rho)H^{-1}(\rho)} \biggr),\]
  where $H^{-1}(\cdot)$ is the inverse of $H(x) = \rho_x/\rho$.
  Furthermore,
  in their proof of Theorem~2,
  they show that for any function $\phi(y)$ such that
  $\phi(y) = o(y^\epsilon)$ for all $\epsilon > 0$,\footnote{%
    While Lin et al.~\cite{Lin} only mention this property for
    the specific case of $\phi(y) = \ln y$,
    their proof easily generalizes.}
  \[\lim_{\rho \to 1} \phi\biggl(\frac{1}{1-\rho} \biggr)(1-\rho)H^{-1}(\rho) = 0.\]
  Applying this result with $\phi(y) = \ln^2(y)$
  proves in case~(ii).
\end{proof}

\section{Additional Simulations}
\label{app:simulation}

We include here some additional simulation results
covering a wider range of cases than Section~\ref{sec:simulation}.
We only show the tightest guardrails ($g = 1$),
as those guardrails generally yield the lowest mean response times.
We vary the parameters of the simulations in from Section~\ref{sec:simulation}
in three different ways:
\begin{itemize}
\item
  adding many more servers
  (Figures~\ref{fig:bp-100-many-pol-g} and~\ref{fig:bm-100-many-pol-g}),
\item
  decreasing load (Figure~\ref{fig:bp-light-many-pol-g}), and
\item
  varying the job size distribution
  (Figures~\ref{fig:hyper-many-pol-g} and~\ref{fig:exp-many-pol-g}).
\end{itemize}
In nearly every case guardrails improve or at least do not degrade
mean response time of the underlying policy.
In particular, as a general rule,
G-LWL is nearly always tied for minimum mean response time
among all dispatching policies simulated.

We omit the FPI heuristic from simulations involving the Bimodal job size distribution
because Hytti\"a et al. \cite{hyytia-fpi} only derive the policy
for continuous job size distributions.
We omit the FPI heuristic from the simulations of other job size distributions
in Figures~\ref{fig:hyper-many-pol-g} and~\ref{fig:exp-many-pol-g}
due to lack of time.

\begin{figure*}
  \begin{subfigure}{.48\linewidth}
    \centering
    \includegraphics[width=\linewidth, trim={0 1em 0 0}]{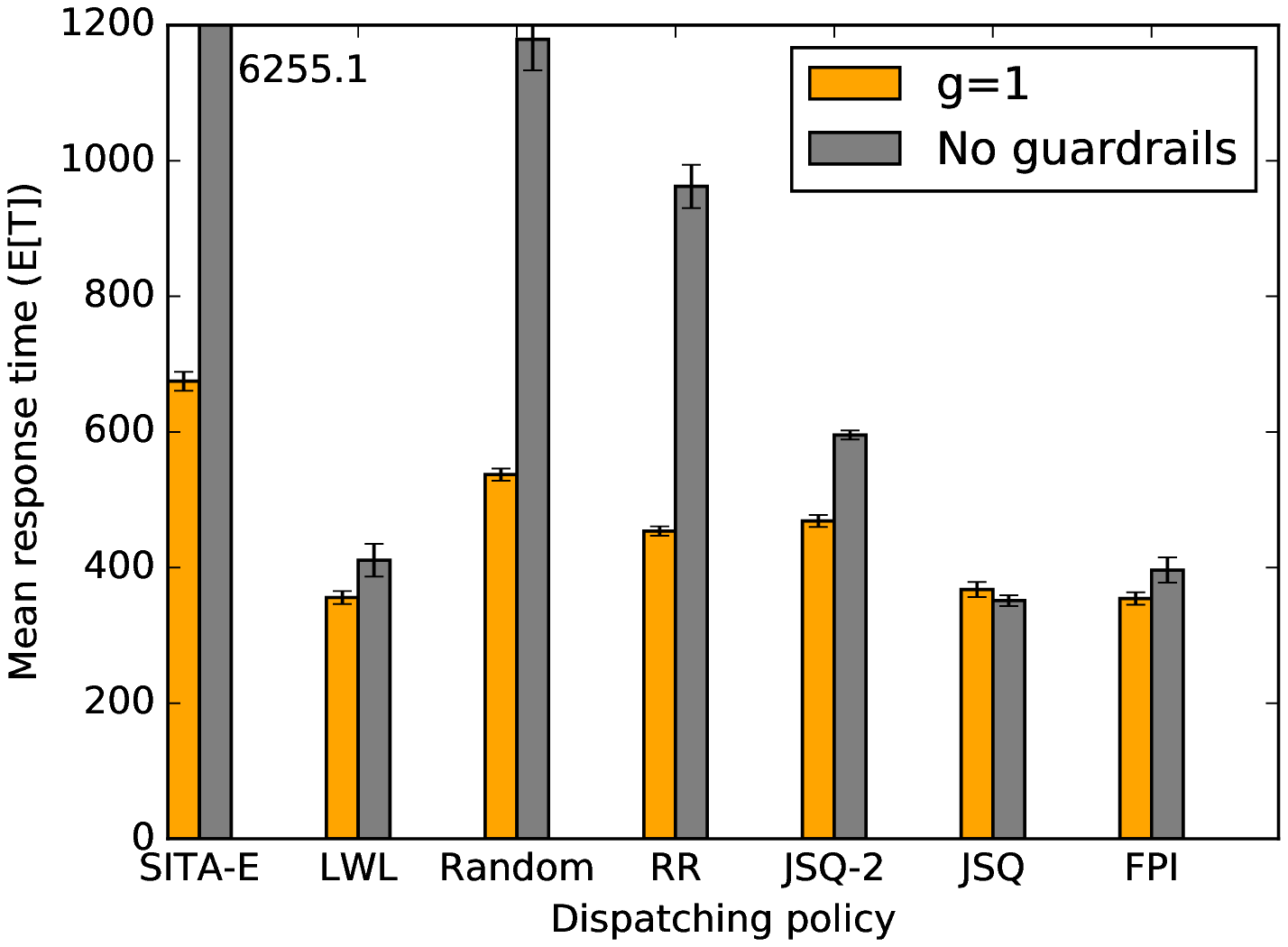}
  \caption{$\rho = 0.98$}
  \vspace{-0.5\baselineskip}
  \end{subfigure}\hfill
  \begin{subfigure}{.48\linewidth}
    \centering
    \includegraphics[width=\linewidth, trim={0 1em 0 0}]{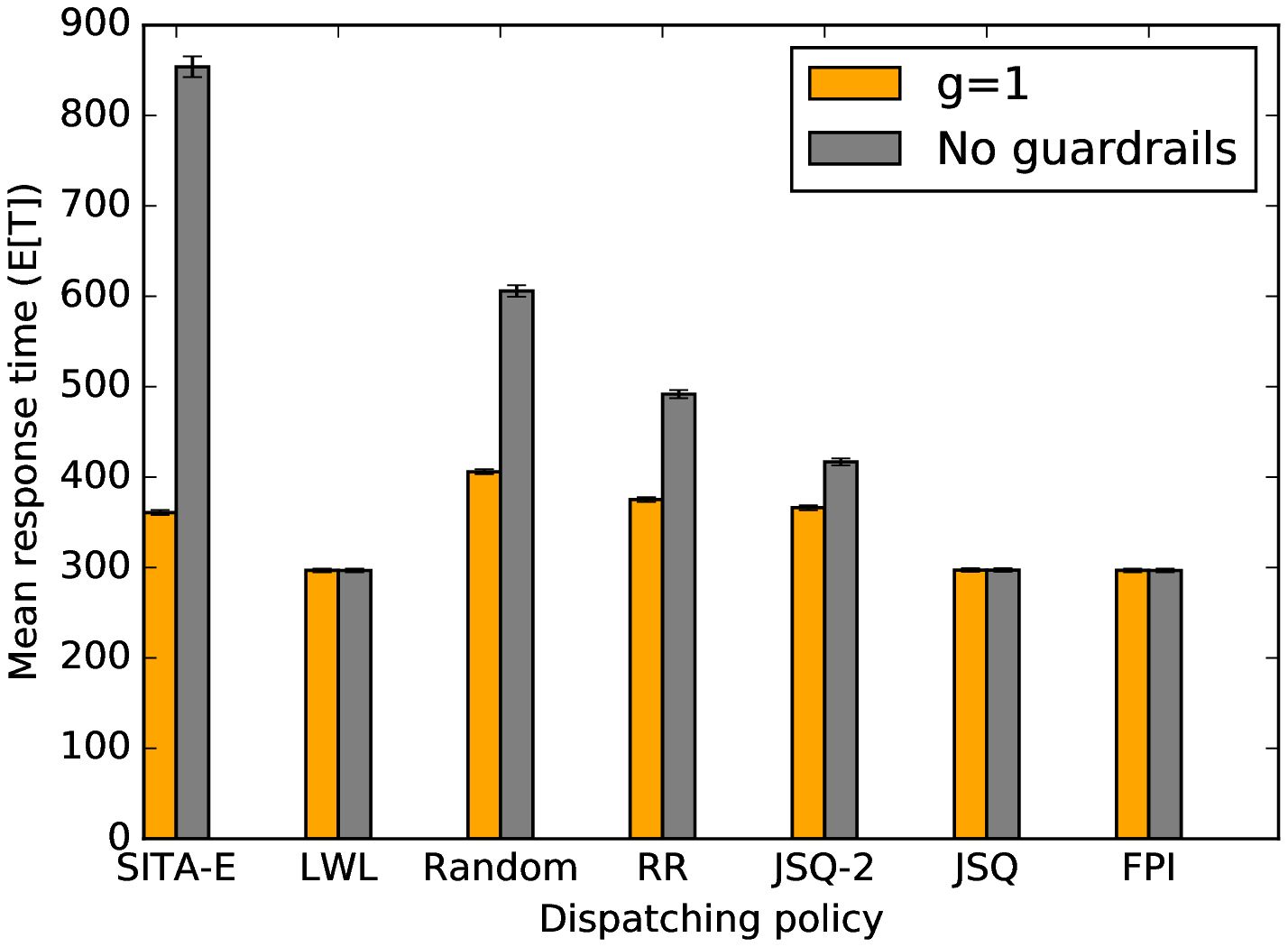}
    \caption{$\rho = 0.80$}
    \vspace{-0.5\baselineskip}
  \end{subfigure}
  \caption{%
    Simulation with many servers: $k=100$.
    Size distribution shown is Bounded Pareto
    with $\alpha=1.5$ and range $[1, 10^6]$.
    $C^2 \sim 333$. 10 trials simulated, 95\% confidence intervals shown.}
  \label{fig:bp-100-many-pol-g}
\end{figure*}

\begin{figure*}
  \begin{subfigure}{.48\linewidth}
    \centering
    \includegraphics[width=\linewidth, trim={0 1em 0 0}]{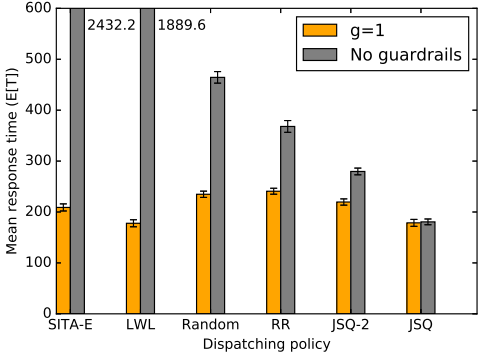}
  \caption{$\rho = 0.98$}
  \vspace{-0.5\baselineskip}
  \end{subfigure}\hfill
  \begin{subfigure}{.48\linewidth}
    \centering
    \includegraphics[width=\linewidth, trim={0 1em 0 0}]{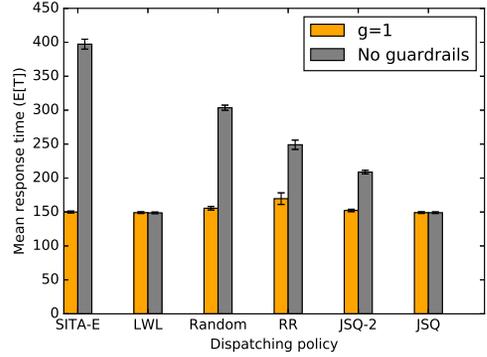}
    \caption{$\rho = 0.80$}
    \vspace{-0.5\baselineskip}
  \end{subfigure}
  \caption{%
      Simulation with many servers: $k=100$.
      Size distribution shown is Bimodal with
      99.95\% size 1 jobs and 0.05\% size 1000 jobs.
      $C^2 \sim 221$. 10 trials simulated, 95\% confidence intervals shown.}
  \label{fig:bm-100-many-pol-g}
\end{figure*}

\begin{figure*}
  \begin{subfigure}{.48\linewidth}
    \centering
    \includegraphics[width=\linewidth, trim={0 1em 0 0}]{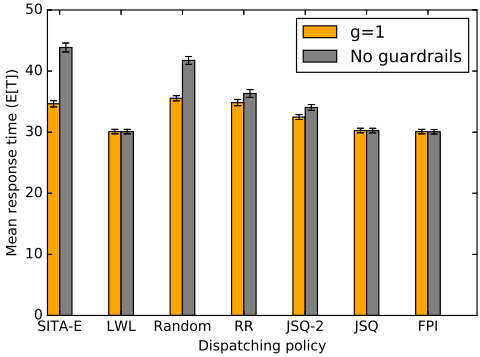}
    \caption{Bounded Pareto job size distribution:
      $\alpha=1.5$, range $[1, 10^6]$,
      $C^2 \sim 333$.}
    \vspace{-0.5\baselineskip}
  \end{subfigure}\hfill
  \begin{subfigure}{.48\linewidth}
    \centering
    \includegraphics[width=\linewidth, trim={0 1em 0 0}]{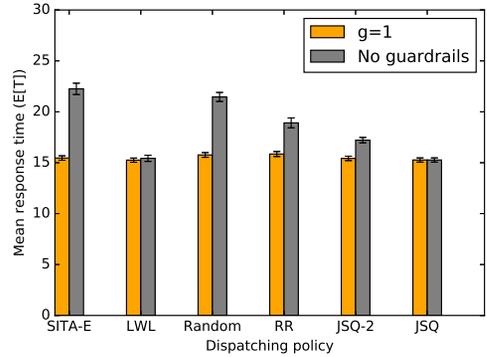}
    \caption{Bimodal job size distribution:
      size 1 w.p. 99.95\%, size 1000 w.p. 0.05\%,
      $C^2 \sim 221$.}
    \vspace{-0.5\baselineskip}
  \end{subfigure}
  \caption{%
    Simulation with light traffic: $\rho = 0.5$.
    Simulation uses $k=10$ servers.
    Two job size distributions shown.
    10 trials simulated, 95\% confidence intervals shown.}
  \label{fig:bp-light-many-pol-g}
\end{figure*}

\begin{figure*}
  \begin{subfigure}{.48\linewidth}
    \centering
    \includegraphics[width=\linewidth, trim={0 1em 0 0}]{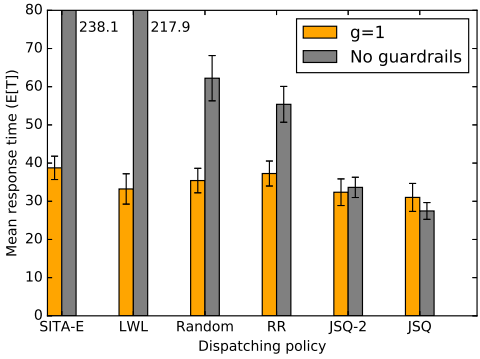}
    \caption{$\rho = 0.98$}
    \vspace{-0.5\baselineskip}
  \end{subfigure}\hfill
  \begin{subfigure}{.48\linewidth}
    \centering
    \includegraphics[width=\linewidth, trim={0 1em 0 0}]{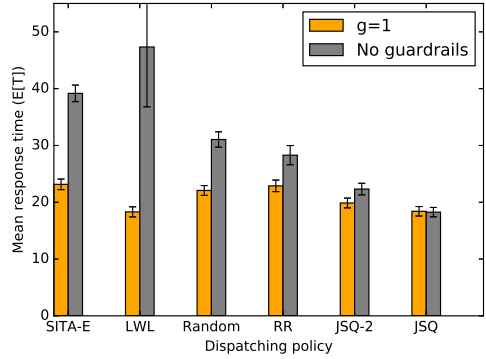}
    \caption{$\rho = 0.80$}
    \vspace{-0.5\baselineskip}
  \end{subfigure}
  \caption{%
    Simulation with different size distribution:
    Hyperexponential with mean $\sim 1.5$ and $C^2 \sim 444$.
    Simulation uses $k=10$ servers.
    10 trials simulated, 95\% confidence intervals shown.}
  \label{fig:hyper-many-pol-g}
\end{figure*}

\begin{figure*}
  \begin{subfigure}{.48\linewidth}
    \centering
    \includegraphics[width=\linewidth, trim={0 1em 0 0}]{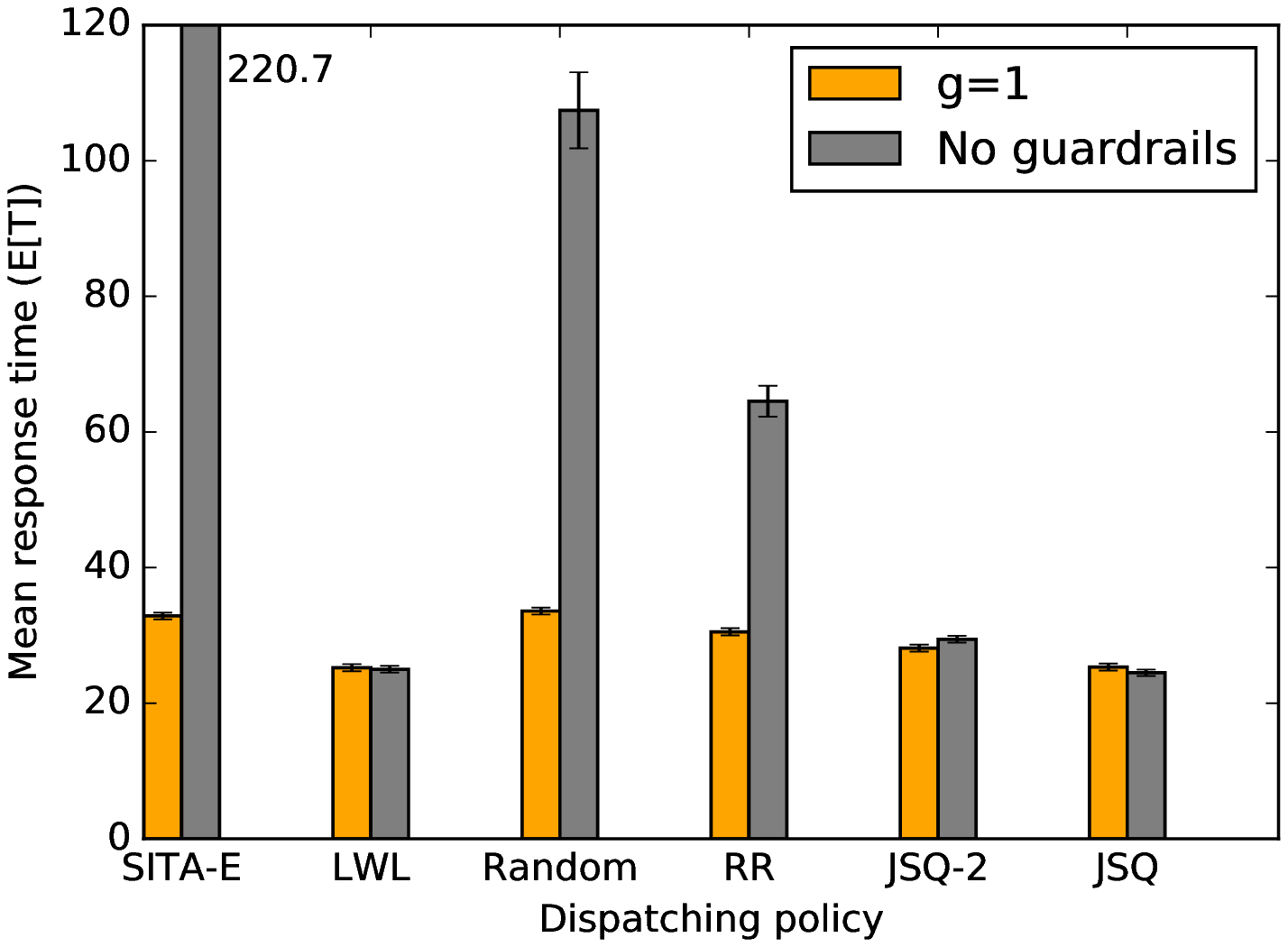}
    \caption{$\rho = 0.98$}
    \vspace{-0.5\baselineskip}
  \end{subfigure}\hfill
  \begin{subfigure}{.48\linewidth}
    \centering
    \includegraphics[width=\linewidth, trim={0 1em 0 0}]{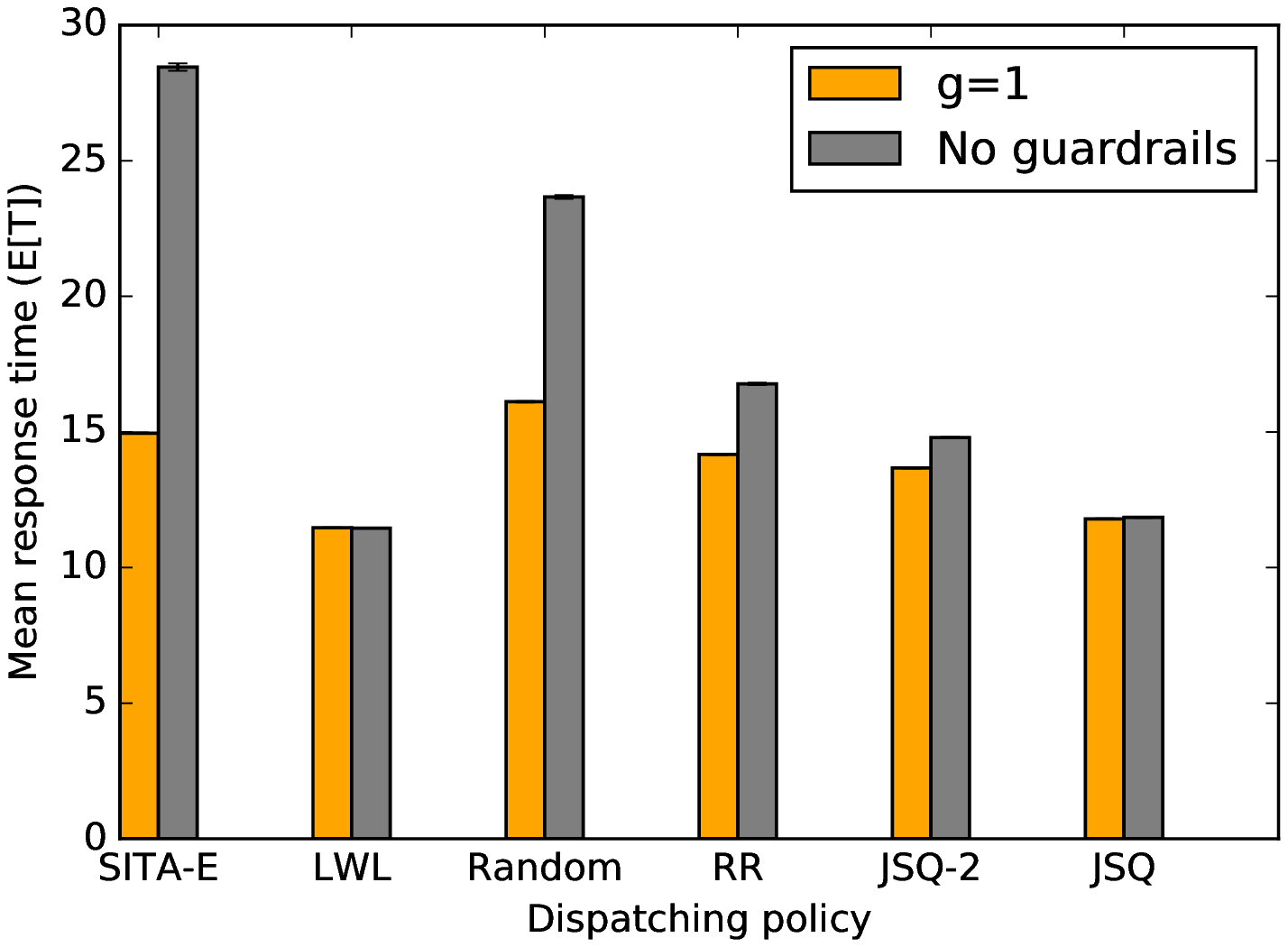}
    \caption{$\rho = 0.80$}
    \vspace{-0.5\baselineskip}
  \end{subfigure}
  \caption{%
    Simulation with different size distribution:
    Exponential with mean $1$.
    Simulation uses $k=10$ servers.
    10 trials simulated, 95\% confidence intervals shown.}
  \label{fig:exp-many-pol-g}
\end{figure*}

Figures~\ref{fig:bp-100-many-pol-g} and~\ref{fig:bm-100-many-pol-g}
show simulations with \emph{many more servers}.
Specifically, they use $k = 100$ servers,
as opposed to $k = 10$ in other simulations.
The only setup where guardrails degrade mean response time
of the underlying policy
is JSQ with Bounded Pareto job size distribution in heavy traffic,
shown in Figure~\ref{fig:bp-100-many-pol-g}~(a),
but G-LWL and G-FPI have performance on par with JSQ in that case.

Figure~\ref{fig:bp-light-many-pol-g}
shows simulations with \emph{light traffic}, specifically $\rho = 0.5$.
The trends are largely the same as those in Section~\ref{sec:simulation},
but the differences in mean response time are smaller.
This is to be expected
because a large fraction of jobs experience no delay.
In fact, the mean response time is nearly equal to the mean service time
for many of the dispatching policies
($E[T] \approx 30$ in (a), $E[T] \approx 15$ in (b)).
Guardrails are particularly effective in the Bimodal case shown in (b),
dispatching nearly every small job to a server with no other small jobs.
In addition to the loads shown,
we have also simulated a range of loads from $\rho = 0.2$ to $\rho = 0.9975$.
The trends are consistent across all loads,
but the differences are less pronounced at lower loads.

Figures~\ref{fig:hyper-many-pol-g} and~\ref{fig:exp-many-pol-g}
show simulations with \emph{different job size distributions}.
We specifically simulate with Hyperexponential and Exponential
job size distributions,
representing another high-variance distribution
and a low-variance distribution, respectively.
In the Hyperexponential case shown in Figure~\ref{fig:hyper-many-pol-g},
LWL performs particularly poorly without guardrails,
similar to the Bimodal case.
Roughly speaking, this is because there again are two types of jobs,
though each has an exponential distribution instead of a deterministic one,
and a job of the large type can cause many jobs of the small type
to be dispatched to a single server.
But again, guardrails effectively mitigate this problem,
although JSQ slightly outperforms G-LWL in very heavy traffic.
In the Exponential case, LWL without guardrails performs well already,
but adding guardrails does not degrade its performance.

\end{document}